\documentclass[journal]{IEEEtran}
\usepackage{array,cite}
\usepackage{graphicx,epsfig,subfigure}
\usepackage{amsthm}
\usepackage{amsmath}
\usepackage{mdwmath}
\usepackage{array}
\usepackage{subeqnarray}
\usepackage{stfloats}
\usepackage{algorithm,algorithmic}
\usepackage{xcolor}
\usepackage{amssymb}
\usepackage{epstopdf}
\usepackage{setspace}
\usepackage{flushend}

\theoremstyle{plain}
\newtheorem{thm}{Theorem}

\theoremstyle{definition}
\newtheorem{defn}{Definition}[]

\theoremstyle{remark}
\newtheorem{rem}{Remark}

\newtheorem{prop}[defn]{Proposition}

\begin{document}

\title{On the Secrecy Unicast Throughput Performance of NOMA Assisted Multicast-Unicast Streaming With Partial Channel Information}
%
%
% author names and IEEE memberships
% note positions of commas and nonbreaking spaces ( ~ ) LaTeX will not break
% a structure at a ~ so this keeps an author's name from being broken across
% two lines.
% use \thanks{} to gain access to the first footnote area
% a separate \thanks must be used for each paragraph as LaTeX2e's \thanks
% was not built to handle multiple paragraphs
%

%\author{Author1,
%        Author2,
%        Author3

\author{Bo~Chen,
        Qingjiang~Shi,
        Yunlong~Cai, and
        Youming~Li
\thanks{B. Chen and Y. Li are with the College of Information and Electronic Engineering, Ningbo University, Ningbo 315211, China (e-mail: \{chenbo3, liyouming\}@nbu.edu.cn).\par
Q. Shi is with the School of Software Engineering \& College of Electronics and Information Engineering, Tongji University, Shanghai 201804, China (e-mail: qing.j.shi@gmail.com).\par
Y. Cai is with the Department of Information Science and Electronic Engineering, Zhejiang University, Hangzhou 310027, China (e-mail: ylcai@zju.edu.cn).
}

\iffalse
\thanks{B. Chen is with the Department
of Information Science and Electronic Engineering, Zhejiang University, Hangzhou, China,
310027 e-mail: 3070902061chenbo@zju.edu.cn.

This work was supported by the National Natural Science Foundation of China under Grant number 91538103, and the 5th Generation Mobile Communication Program in China (863 Project) under Grant number 2012AA01A507, 2014AA01A707.}
\fi

%\thanks{J. Doe and J. Doe are with Anonymous University.}% <-this % stops a space
%\thanks{Manuscript received April 19, 2005; revised September 17, 2014.}
}

\maketitle

% As a general rule, do not put math, special symbols or citations
% in the abstract or keywords.
\begin{abstract}
This paper considers a downlink single-cell non-orthogonal multiple access (NOMA) network with uniformly deployed users, while a mixed multicast and unicast traffic scenario is taken into account. By guaranteeing the quality of service (QoS) of the multicast traffic, the multicast outage and secrecy unicast throughput performance is evaluated. In particular, two types of partial channel state information (CSI), namely the imperfect CSI and CSI based on second order statistics (SOS), are investigated in the analytical framework. For both two cases, the closed-form approximations for the multicast outage probability and secrecy unicast throughput are derived except that the approximation for the secrecy unicast throughput in SOS-based CSI only concerns two users. As to the multicast outage probability, the simulation results demonstrate that the NOMA scheme considered in both two cases achieves superior performance compared to the traditional orthogonal multiple access (OMA) scheme, while for the secrecy unicast throughput, the NOMA scheme shows great advantages over the OMA scheme in good condition (high signal-to-noise ratio (SNR)), but is inferior to the OMA scheme in poor condition (low SNR). Moreover, both two cases achieve similar performance except that the NOMA scheme with imperfect CSI obtains larger secrecy unicast throughput than that based on SOS under high SNR condition. Finally, the provided numerical results also confirm that the derived approximations of both two cases for the multicast outage probability and secrecy unicast throughput match well with the Monte Carlo simulations.
\end{abstract}

% Note that keywords are not normally used for peerreview papers.
\begin{IEEEkeywords}
Non-orthogonal multiple access, mixed multicast and unicast, imperfect channel state information, second order statistics, secrecy unicast throughput.
\end{IEEEkeywords}

% For peer review papers, you can put extra information on the cover
% page as needed:
% \ifCLASSOPTIONpeerreview
% \begin{center} \bfseries EDICS Category: 3-BBND \end{center}
% \fi
%
% For peerreview papers, this IEEEtran command inserts a page break and
% creates the second title. It will be ignored for other modes.
\IEEEpeerreviewmaketitle

\section{Introduction}
% The very first letter is a 2 line initial drop letter followed
% by the rest of the first word in caps.
%
% form to use if the first word consists of a single letter:
% \IEEEPARstart{A}{demo} file is ....
%
% form to use if you need the single drop letter followed by
% normal text (unknown if ever used by IEEE):
% \IEEEPARstart{A}{}demo file is ....
%
% Some journals put the first two words in caps:
% \IEEEPARstart{T}{his demo} file is ....
%
% Here we have the typical use of a "T" for an initial drop letter
% and "HIS" in caps to complete the first word.
\IEEEPARstart{N}{on-orthogonal} multiple access (NOMA) has been recognized as a promising technology to satisfy the challenging requirements of the fifth generation (5G) wireless networks, such as high data speed, massive connectivity, and low latency\cite{Islam,Ding}. On the other hand, downlink NOMA has been introduced in 3rd generation partnership project (3GPP)-long term evolution-advance (LTE-A) systems\cite{5G}, aiming to improve the spectral efficiency. The principle behind NOMA is that multiple users' signals are superimposed at the base station (BS) with different power level, and the user decodes its own signal by removing signals of poor channel conditions' users, which is called successive interference cancellation (SIC) technique\cite{Saito}. Note that the key feature of the downlink NOMA is to take user fairness into consideration as it allocates more power to users with worse channel conditions than those with better channel conditions, which realizes an improved trade-off between user fairness and system throughput.

Based on the perfect channel state information (CSI), performance analysis of the downlink NOMA is extensively studied in\cite{Ding1,Shi,Yang,Lv,Zeng}. The outage and sum-rate performance of NOMA in a cellular downlink scenario with randomly deployed users was analyzed in \cite{Ding1}. Meanwhile, \cite{Shi} investigated the outage balancing problem while the issues of power allocation, decoding order selection, and user grouping were taken into account. As to the downlink and uplink NOMA scenarios with two users for flexible quality of service (QoS) requirements, \cite{Yang} developed a novel power allocation scheme and established the exact expressions of the outage probability and average rate. As research goes deep, by taking the cooperative communication into consideration, the closed-form expressions for the outage probability and ergodic sum rate were derived in a downlink NOMA system with cooperative half-duplex relaying\cite{Lv}. Finally, in a multiple-input multiple output (MIMO) system, \cite{Zeng} studied sum channel capacity and ergodic sum capacity performance of NOMA in detail.

However, in practice, perfect CSI at the transmitter is usually not available, because obtaining the perfect CSI consumes a significant system overhead, especially for the wireless network with a large number of users. Furthermore, one of key features towards 5G is highly mobile, leading to rapidly changing channel, which makes it greatly challenging to achieve perfect CSI as the transmitter. Therefore, assuming partial CSI, several related works have been investigated in\cite{Xu,Timotheou,Sun,Yang1} recently. In \cite{Xu}, the outage performance of the downlink NOMA was studied for the case where each user feeds back only one bit of its CSI to the BS. Based on average CSI, optimal power allocation was analyzed while fairness for the downlink users were ensured in\cite{Timotheou}. Meanwhile, \cite{Sun} solved the ergodic capacity maximization problem of a MIMO-NOMA system with second order statistical (SOS) CSI at the transmitter. Emphasize that both \cite{Timotheou} and \cite{Sun} focused on the case where user location are fixed, namely the distances and path loss are deterministic. Therefore, considering a downlink single-cell NOMA network with uniformly deployed users, an analytical framework to evaluate the outage and sum rate performance was developed in\cite{Yang1}.

Note that all the works mentioned above considered only single unicast traffic. Meanwhile, Several works have studied the downlink NOMA with a mixed multicast and unicast traffic scenario\cite{Choi,Ding2,Yang2}. In \cite{Choi}, Multicast beamforming with superposition coding (SC) was investigated for multi-resolution broadcast where both data streams of high priority (HP) and low priority (LP) were to be transmitted for a near user, while only data stream of LP was to be transmitted to a far user. \cite{Ding2} developed a novel beamforming and power allocation scheme while the unicast performance was improved and the reception reliability of the multicast was maintained. What's more, \cite{Ding2} also analyzed how the use of NOMA can prevent those multicast receivers intercepting the unicast messages and proposed the secrecy unicast rate metric. Finally, \cite{Yang2} introduced a novel NOMA unicast-multicast system, where a number of unicast users and a group of multicast users shared the same wireless resource.

However, to date few work has made performance analysis where both partial CSI condition and mixed unicast-multicast traffic are taken into account. Hence, in this paper, we consider a downlink single-cell NOMA network, where the users are uniformly distributed in a disk and the BS is located at the center. Meanwhile a mixed multicast and unicast traffic scenario is taken into account, where the BS transmits two types of data streams, one for multicast and one for unicast. Assuming partial CSI at the BS, by guaranteeing the QoS of the multicast traffic, the multicast outage and secrecy unicast throughput performance are investigated. In particular, similar to \cite{Yang1}, we consider two types of partial CSI, namely imperfect CSI and SOS-based CSI, which are defined as follows.

\begin{itemize}
\item \emph{Imperfect CSI}: According to \cite{Yang1,Yoo,Han}, a channel estimation error model is assumed, where the BS and the users estimate the channel with a priori knowledge of the variance of estimation error. More specifically, we concentrate on the minimum mean square error (MMSE) channel estimation error model\cite{Ikki,Wang,Ma}.
\item \emph{SOS-based CSI}: Only the distances between the BS and the users are known at the BS. The motivation for studying such SOS-based CSI is analyzed below. First, distance information is easy to obtain in practical communications. Second, small-scaling fading only weekly changes the large-scale fading, which means that large-scale fading is dominant in CSI\cite{Cline}.
\end{itemize}

In particular, the contribution of this paper is three-fold:
\begin{itemize}
\item We develop a novel performance analysis framework in a downlink single-cell NOMA network with uniformly deployed users, where only partial CSI is available at the BS and a mix multicast-unicast traffic is considered. We guarantee the QoS of the multicast, which means that the minimum rate of the multicast message is ensured, and the corresponding outage performance is studied. Meanwhile, since the unicast message is broadcast to all the users, the secrecy unicast throughput is also investigated, where the power allocation and unicast user selection problem is properly formulated and mathematically solved.
\item Two types of partial CSI are considered, namely imperfect CSI and SOS-based CSI. For both two cases, employing the probability theory, order statistics theory and multi-binomial theorem, the closed-form approximate expressions for the outage probability and secrecy unicast throughput are derived. Note that as to the SOS-based CSI case, due to significant complication of calculating the secrecy unicast throughput , we only focus on two users. Numerical results are provided to demonstrate that the derived approximate expressions of both two cases for the outage probability and secrecy unicast throughput match well with the Monte Carlo simulations.
\item The corresponding orthogonal multiple access (OMA) scheme with partial CSI is developed and a performance comparison between the NOMA and OMA schemes is made. Simulation results show that as to the multicast outage probability, the NOMA scheme in both two cases achieves superior performance compared to the OMA scheme, while for the secrecy unicast throughput, the NOMA scheme shows great advantages over the OMA scheme in good condition (high signal-to-noise ratio (SNR)), but is slightly inferior to the OMA scheme in bad condition (low SNR). Moreover, a comparison between the NOMA scheme with imperfect CSI and the NOMA scheme with SOS-based CSI is presented. Numerical results suggest that both two NOMA schemes obtain similar performance except that the NOMA scheme with imperfect CSI achieves larger secrecy unicast throughput than that based on SOS in high SNR.
\end{itemize}

The rest of this paper is organized as follows. In Section \ref{sec:sysmod}, we introduce the system model and formulate the corresponding power allocation problem. In Section \ref{sec:per_analy}, we first mathematically solve the resource optimization problem. Second, by assuming two types of partial CSI, the multicast outage and secrecy unicast throughput performance are analyzed in detail. Third, the corresponding OMA scheme with partial CSI is provided as a benchmark. Section \ref{sec:rest} presents the simulation results and the conclusions are drawn in Section \ref{sec:conl}.

\section{System Model and Problem Formulation}\label{sec:sysmod}
In this section, we consider a downlink single-cell NOMA network with uniformly deployed users while a mixed multicast and unicast traffic scenario is takeing into account. Aiming to maximize the secrecy unicast throughput, the power allocation and unicast user selection problem is correspondingly formulated.

\subsection{System Model}\label{sub:sysmod}
Consider a single-cell downlink wireless network with one BS communicating with $K$ users. The location of $K$ users are uniformly distributed in a disc with radius $D$, which is denoted as $\cal{D}$, while the BS is located in the center of the disc. Here all the users and the BS are equipped with single antenna. Furthermore, we assume that all users share the same wireless channel resource. The channel between the BS and user $U_k$, $k \in {\cal{K}}=\{1,2,\cdots,K\}$, is modeled as $h_k = g_k d_k^{-\frac{\eta}{2}}$, where $g_k$ is the Rayleigh fading coefficient, $d_k$ is the distance between BS and user $U_k$, and $\eta$ is the path loss exponent. Here we assume that $g_k$ ($k \in \cal{K}$) follows complex normal distribution with zero mean and unit variance, that is $g_k \sim {\cal{CN}}(0,1)$. Denote $\alpha_k$ as the channel gain, namely $\alpha_k = |h_k|^2$. Furthermore, the background noise is modeled as the additive white Gaussian noise ({\bf{AWGN}}) with zero mean and variance $\sigma^2$.

Note that the optimal performance can be achieved with perfect CSI. However, perfect CSI is usually not available in practical communications as the significant consumed overhead and great challenge in achieving CSI exactly. Motivated by this fact, this paper consider two types of practical channel models: imperfect CSI and SOS-based CSI, which will be fully discussed in Section \ref{sec:per_analy}.

This work focuses on a mixed multicast and unicast traffic scenario, i.e., the BS has two messages to send. The multicast traffic is to be received by all the users, while the unicast traffic is intended to a particular user. Here we assume that the BS has unlimited multicast data to transmit, and infinite unicast data for each user to send. For a particular time slot, based on the partial CSI, the BS can select a specific user to transmit the unicast traffic, while the secrecy unicast throughput is maximized and the QoS of the multicast traffic is guaranteed.
\subsection{Problem Formulation}
The transmit data combines the multicast and unicast streaming. The multicast message $s_M$ is intended to all the users, whereas the unicast message $s_U$ is to be received by a particular user. Then, employing the NOMA transmission technique, the transmit data ${x}$, sent by the BS, can be denoted as
\begin{equation}\label{1.1}
x = \sqrt{P_T} (\sqrt{\theta_M} s_M + \sqrt{\theta_U} s_U),
\end{equation}
where $P_T$ is the total transmit power, $\theta_M$ and $\theta_U$ are the power allocation coefficients for the multicast and unicast messages, respectively. Note that $\theta_M$ and $\theta_U$ are designed to satisfy
\begin{equation}\label{1.2}
\theta_M + \theta_U \le 1, \theta_M \ge 0, \theta_U \ge 0.
\end{equation}

The received data at user $U_k$ is
\begin{equation}\label{1.3}
y_k = h_k x + n_k.
\end{equation}
According to the principle of the NOMA, here $U_k$ detects $s_M$ by treating $s_U$ as noise. Therefore, the signal-interference-noise ratio (SINR) for decoding $s_M$ at $U_k$ can be calculated as
\begin{equation}\label{1.4}
SINR_k^M = \frac {\theta_M \alpha_k}{\theta_U \alpha_k+\frac{1}{\rho}},
\end{equation}
where $\rho$ denotes the transmit SNR, namely $\rho = \frac{P_T}{\sigma^2}$. Once $s_M$ is successfully decoded, SIC will be carried out at $U_k$ and $s_M$ is totally removed for detecting $s_U$. Hence, the SNR for decoding $s_U$ is given by
\begin{equation}\label{1.5}
SNR_k^U = \rho \theta_U \alpha_k.
\end{equation}

According to the shannon capacity theory, the achievable rate of the multicast and unicast streams at $U_k$ can be obtained as
\begin{equation}\label{1.6}
R_k^M = \log_2 {(1 + SINR_k^M)}, R_k^U = \log_2 {(1 + SNR_k^U)},
\end{equation}
where the transmit bandwidth is normalized here.

The optimization problem can be mathematically formulated as follows. First, the QoS of the multicast stream is guaranteed, namely the minimum rate of the multicast traffic is ensured, that is
\begin{equation}\label{1.7}
R_k^M \ge R_M, \forall k \in \mathcal{K},
\end{equation}
where $R_M$ is the minimum rate of the multicast traffic.

Similar to \cite{Ding2}, the secrecy unicast throughput is defined as
\begin{equation}\label{1.8}
R_S^U \triangleq [R_j^U- \max{\{R_k^U, k \in \mathcal{K} \backslash \{j\} \}}]^+ ,
\end{equation}
where $[x]^+ \triangleq \max \{0,x\}$. Note that for a particular time slot, unicast traffic receiver $U_j$ should be carefully selected as discussed above.

Therefore, the power allocation and unicast user selection problem considered in this paper can be written as
\begin{equation}\label{1.9}
\max_{j,\theta_M,\theta_U} R_S^U,
\end{equation}
subject to
\begin{equation}\label{1.10}
(\ref{1.2}), (\ref{1.7}).
\end{equation}

\section{Performance analysis}\label{sec:per_analy}
In this section, the analytical framework to evaluate the multicast outage and secrecy unicast throughput performance is developed. First, the optimization problem (\ref{1.9}) is fully solved. Second, based on two types of partial CSI, namely imperfect CSI and SOS-based CSI, the closed-form approximate expressions for the multicast outage probability and secrecy unicast throughput are mathematically derived. Third, the corresponding OMA scheme with partial CSI is provided as a benchmark.

\subsection{Solution of problem (\ref{1.9})}\label{subsec:pro}
Here we denote $\pi = \{\pi_1, \pi_2, \cdots, \pi_K\}$ as a permutation of the user indices, and if $\pi_i = k$, then the user $U_k$ has the $i$-th best channel gain. In other words, the channel gains are sorted as $\alpha_{\pi_1} \ge \alpha_{\pi_2} \ge \cdots \ge \alpha_{\pi_K}$.

According to the definition of $R_S^U$ in (\ref{1.8}), it is easy to derive that $R_S^U = [\log_2 {(\frac{1+\rho \theta_U \alpha_j}{1+ \rho \theta_U \max_{k\in \mathcal{K}\backslash \{j\} } \alpha_k})}]^+$. As $\alpha_{\pi_1} \ge \alpha_{\pi_2} \ge \cdots \ge \alpha_{\pi_K}$, in order to maximize the $R_S^U$, the BS should choose $j$ as user $\pi_1$. Then $R_S^U = \log_2 {(\frac{1+\rho \theta_U \alpha_{\pi_1}}{1+ \rho \theta_U \alpha_{\pi_2}})} = \log_2 {( \frac{\alpha_{\pi_1}}{\alpha_{\pi_2}} - \frac{\frac{\alpha_{\pi_1}}{\alpha_{\pi_2}}-1}{1+\rho \theta_U \alpha_{\pi_2}} )}$. Therefore, $R_S^U$ increases with $\theta_U$, which accords with common sense. Furthermore, for the given $\theta_M$ and $\theta_U$, it is easy to find that $R_{\pi_i}^M > R_{\pi_j}^M$ if $i>j$. Thus constraint (\ref{1.7}) can be simplified as
\begin{equation}\label{}\label{2.1}
R_{\pi_K}^M \ge R_M.
\end{equation}

Hence, Problem (\ref{1.9}) can be equally expressed as
\begin{equation}\label{eq:a1}
\max_{\theta_M,\theta_U} R_S^U = \log_2 {(\frac{1+\rho \theta_U \alpha_{\pi_1}}{1+ \rho \theta_U \alpha_{\pi_2}})},
\end{equation}
subject to
\begin{equation}\label{eq:a2}
(\ref{1.2}), (\ref{2.1}).
\end{equation}

\begin{thm}\label{thm:1}
The optimal solution of problem ({\ref{eq:a1}}) can be obtained while both (\ref{1.2}) and (\ref{2.1}) achieve equality.
\end{thm}
\begin{proof}\label{prof:1}
See Appendix \ref{appe:1}.
\end{proof}

Based on \emph{Theorem \ref{thm:1}}, the optimal solution $(\theta_M^*,\theta_U^*)$ should satisfy
\begin{equation}\label{eq:a3}
\begin{aligned}
& \theta_M^* + \theta_U^* = 1 \cr
& R_{\pi_K}^M(\theta_M^*,\theta_U^*) = R_M.
\end{aligned}
\end{equation}
Therefore, we can obtain $\theta_U^* = [\frac{\alpha_{\pi_K}- \frac{\epsilon_M}{\rho}}{\alpha_{\pi_K}(1+\epsilon_M)}]^+$, where $\epsilon_M = 2^{R_M}-1$.

\begin{rem}\label{rem:1}
If $\alpha_{\pi_K} < \frac{\epsilon_M}{\rho}$, a multicast outage happens, that is the QoS of the multicast traffic can not be satisfied even if all the power is allocated to $s_M$.
\end{rem}

In conclusion, the optimal solution of problem (\ref{1.9}) is $j = \pi_1, \theta_U^* = [\frac{\alpha_{\pi_K}- \frac{\epsilon_M}{\rho}}{\alpha_{\pi_K}(1+\epsilon_M)}]^+, \theta_M^* =1-\theta_U^* = \min (1,\frac{\alpha_{\pi_K}\epsilon_M+ \frac{\epsilon_M}{\rho}}{\alpha_{\pi_K}(1+\epsilon_M)})$.

\subsection{Performance analysis for imperfect CSI}\label{subsec:imperfect}
In this subsection, based on the imperfect CSI, we analyze the multicast outage and secrecy unicast throughput performance in detail.

Here we assume that the channel feedback to the transmitter is instantaneous and error free, which means that CSI is also achievable at the transmitter whatever CSI the receiver has. As discussed above, such assumption is widely applied in the literature\cite{Yang1,Yoo,Han,Ikki,Wang,Ma}. Define $\hat{h}_k$ as the estimation for channel $h_k$, and estimated channel gain $\hat{\alpha}_k$ can be correspondingly calculated as $\hat{\alpha}_k = |\hat{h}_k|^2$. Similar to \cite{Ikki,Wang,Ma}, assuming the MMSE estimation error, it holds that
\begin{equation}
\alpha_k = \hat{\alpha}_k + \zeta,
\end{equation}
where $\zeta$ is the channel estimation error, which follows a complex Gaussian distribution with mean $0$ and variance $\sigma_\zeta^2$, namely $\zeta \sim {\cal{CN}}(0,\sigma_\zeta^2)$. Thus the estimated channel gain $\hat{\alpha}_k$  follows a complex Gaussian distribution with mean $0$ and variance $\sigma_{\hat{\alpha}_k}^2 = d_k^{-\eta}-\sigma_\zeta^2$\cite{Ikki,Wang,Ma}. Note that $\sigma_\zeta^2$ indicates the quality of channel estimation.

Correspondingly, the optimal solution of problem (\ref{1.9}) is $j = \pi_1, \theta_U^* = [\frac{\hat{\alpha}_{\pi_K}- \frac{\epsilon_M}{\rho}}{\hat{\alpha}_{\pi_K}(1+\epsilon_M)}]^+, \theta_M^* =1-\theta_U^* = \min (1,\frac{\hat{\alpha}_{\pi_K}\epsilon_M+ \frac{\epsilon_M}{\rho}}{\hat{\alpha}_{\pi_K}(1+\epsilon_M)})$.

Next, the multicast outage probability $P_{out}^{1} = Pr(\hat{\alpha}_{\pi_K} < \frac{\epsilon_M}{\rho})$ and statistical expectation of secrecy unicast throughput $R_{S,1}^U = E[R_S^{U}]$ will be analyzed as follows.
\begin{thm}\label{thm:2}
As to $P_{out}^{1}$, it can be approximated as
\begin{equation}\label{Q3:a5}\small
\begin{aligned}
P_{out}^1 \approx 1- [\frac{\pi}{cD} \sum \limits_{i=1}^{c} |\sin{\frac{2i-1}{2c}\pi}| x_i  \exp(-\frac{\epsilon_M }{\rho(x_i^{-\eta}-\sigma_\zeta^2)})]^K,
\end{aligned}
\end{equation}
where $x_i = \frac{D}{2} (1+\cos{(\frac{2i-1}{2c}\pi)})$, and $c$ is the number of terms included in the summation, which controls the approximation accuracy, due to the use of  Gauss-Chebyshev integration \cite{Hildebrand}.
\end{thm}
\begin{proof}\label{prof:2}
See Appendix \ref{appe:2}.
\end{proof}

\begin{rem}\label{rem:2}
Here we evaluate the computational complexity as the number of loops in computation. Then it is easy to obtain that the complexity of $P_{out}^1$ is $O(c)$.
\end{rem}

\begin{rem}\label{rem:3}
The multicast outage expression (\ref{Q3:a5}) is accurate for whole SNR.  Furthermore, $c$ is a important parameter, which affects the accuracy of analytical results. However, as shown by the simulation results, (\ref{Q3:a5}) achieves an accurate approximation even with a small $c$ ($c=50$).
\end{rem}

\begin{rem}\label{rem:4}
Under high SNR, namely $\rho$ is large enough,  it is easy to see that $\exp{(-\frac{\epsilon_M}{\rho (x_i^{-\eta}-\sigma_\zeta^2)})} \approx 1$, then
\begin{equation}\label{Q3:a6}
(P_{out}^1)^\infty \approx 1- [\frac{\pi}{cD} \sum \limits_{i=1}^{c} |\sin{\frac{2i-1}{2c}\pi}| x_i  ]^K.
\end{equation}
According to (\ref{Q3:a6}), we can find that there is no diversity gain, which is in line with our system design, that is all the users independently decode the $s_M$ and $s_U$.
\end{rem}

In addition, the case with perfect CSI is also worth studying as it provides a multicast outage performance upper bound, where the loss due to imperfect CSI can be clearly demonstrated. Fortunately, if perfect CSI is available, namely $\sigma^2_\zeta = 0$, an exact closed-form expression for the multicast outage performance can be derived as follows.
\begin{prop}\label{pro:1}
For the case with perfect CSI, namely $\sigma^2_\zeta = 0$, the multicast outage probability $(P_{out}^1)_{\sigma^2_\zeta = 0}$ can be exactly obtained as
\begin{equation}\label{pro:eq:1}
(P_{out}^1)_{\sigma^2_\zeta = 0} = 1-[\frac{2}{\eta (\frac{\epsilon_M}{\rho})^{\frac{2}{\eta}}D^2}\gamma(\frac{2}{\eta},\frac{\epsilon_M D^\eta}{\rho})]^K,
\end{equation}
where $\gamma(a,b) = \int_0^b{t^{a-1}e^{-t}}dt$ is a lower incomplete gamma function.
\end{prop}
\begin{proof}\label{prof:2.1}
See Appendix \ref{appe:2.1}.
\end{proof}

\begin{rem}\label{rem:4.2}
For the perfect CSI, the complexity of obtaining the exact closed-form expression for $P_{out}^1$ in (\ref{pro:eq:1}) is $O(K)$, which is acceptable for large $K$.
\end{rem}

\begin{rem}\label{rem:4.5}
According to $8.354.1$ in\cite{Zwillinger}, namely $\gamma(\alpha,x) = \sum\limits_{n=0}^\infty \frac{(-1)^n x^{\alpha+n}}{n!(\alpha+n)} $, when $\rho \to \infty$, $(P_{out}^1)_{\sigma^2_\zeta = 0}$ can be expressed as
\begin{equation}\label{eq:rem4.5}\footnotesize
(P_{out}^1)_{\sigma^2_\zeta = 0}^\infty \approx 1 - W(\eta),
\end{equation}
where $W(\eta)$ is a positive constant related to $\eta$. From (\ref{eq:rem4.5}), we can conclude that even for the perfect CSI, there is also no diversity gain, which accords with our system design, that is all the users independently decode the $s_M$ and $s_U$.
\end{rem}

\begin{thm}\label{thm:3}
As to the average secrecy unicast throughput $R_{S,1}^{U}$, it can be approximated as
\begin{equation}\label{Q3:14}\small
\begin{aligned}
R_{S,1}^U &\approx (1-P_{out}^1) \frac{K \pi \rho}{2m\ln 2}  \sum_{u=1}^{m} |\sin \frac{2u-1}{2m}\pi|\cr
&\left\{\frac{1}{\rho\tau_u} - \sum_{r_0+r_1+ \cdots +r_n = K-1} A(r_1,\cdots,r_n) H(r_1,\cdots,r_n,u) \right\}, \cr
\end{aligned}
\end{equation}
where
\begin{equation*}\label{thm3:add}\footnotesize
\begin{aligned}
&A(r_1,\cdots,r_n) = \frac{(K-1)!}{(K-1-\sum_{t=1}^n r_t)!r_1!\cdots r_n!} (\frac{-\pi}{nD})^{r_1+r_2+\cdots + r_n}\cr
&\prod_{t=1}^{n} (|\sin{\frac{2t-1}{2n}\pi}| x_t)^{r_t},\cr
\end{aligned}
\end{equation*}
\begin{equation*}\label{thm3:addd2}\footnotesize
\begin{aligned}
&\tau_u = \frac{1}{2}(\cos \frac{2u-1}{2m} \pi +1),x_t = \frac{D}{2} (1+\cos{\frac{2t-1}{2n}\pi)},\cr
\end{aligned}
\end{equation*}
\begin{equation*}\label{thm3:add1}\footnotesize
\begin{aligned}
&H(r_1,\cdots,r_n,u)= -\frac{\pi}{nD\rho \tau_u} \sum \limits_{t=1}^{n} |\sin{\frac{2t-1}{2n}\pi}| x_t G(r_1,\cdots,r_n,u,t)\cr
&+I_2(r_1,\cdots,r_n),\cr
\end{aligned}
\end{equation*}
\begin{equation*}\label{thm3:add2}\footnotesize
\begin{aligned}
&I_2(r_1,\cdots,r_n) =
\begin{cases}
\frac{1}{\rho\tau_u}, &r_1=r_2=\cdots=r_n=0\cr
0,&\text{othwise}.
\end{cases}
\cr
\end{aligned}
\end{equation*}
\begin{equation*}\label{thm3:add3}\footnotesize
\begin{aligned}
&G(r_1,\cdots,r_n,u,t)= 1-\frac{B(r_1,\cdots,r_n,u)}{\rho \tau_u \bar{\mu}_1}+\cr
&\nu e^{\nu \bar{\mu}_1}\text{Ei}(-\nu \bar{\mu}_1)[\bar{\mu}_1-\frac{B(r_1,\cdots,r_n,u)}{\rho\tau_u} ],\cr
\end{aligned}
\end{equation*}
\begin{equation*}\label{thm3:addd1}\footnotesize
\begin{aligned}
&B(r_1,\cdots,r_n,u) = \tau_u (\sum_{t=1}^n \frac{r_t} {x_t^{-\eta}-\sigma_\zeta^2} ),\cr
\end{aligned}
\end{equation*}
\begin{equation}\label{thm3:add4}\footnotesize
\begin{aligned}
&\bar{\mu}_1 \mathop =\limits^{\triangle} \mu_1(r_1,\cdots,r_n,u,t) = \frac{B(r_1,\cdots,r_n,u)+\frac{1}{x_t^{-\eta}-\sigma_\zeta^2}}{\rho \tau_u},\cr
&\nu = 1+\epsilon_M .
\end{aligned}
\end{equation}
Note that $\text{Ei}(x) = \int_{-\infty}^x \frac{e^t}{t}\text{d}t, \text{for } x<0$, and $P_{out}^1$ has been derived in (\ref{Q3:a5}). Similarly, $m$ and $n$ controls the approximation accuracy due to the use of Gauss-Chebyshev integration\cite{Hildebrand}.
\end{thm}
\begin{proof}\label{prof:3}
See Appendix \ref{appe:3}.
\end{proof}

\begin{rem}\label{rem:5}
The computational complexity of $R_{S,1}^U$ can be analyzed as follows. Note that the complexity of calculating $P_{out}^1$ is $O(c)$ and the possibilities of $r_0+r_1+\cdots +r_n = K-1$ is $\binom{K+n-1}{n}$. Next, the complexity of calculating $A(r_1,\cdots,r_n)$ is $O(n)$, while that of obtaining $H(r_1,\cdots,r_n,u)$ is $O(n^2)$. Hence, the complexity of $R_{S,1}^U$ can be easily achieved as $O\left(c+m\binom{K+n-1}{n}(n+n^2)\right)$. Here an important tradeoff between the accuracy and computational complexity should be carefully considered. Note that the complexity dramatically increases with $n$ due to $n^2\binom{K+n-1}{n}$. However, a small $n$ will lead to a large estimation gap. According to our simulation results, $n=10$ is proper as both the complexity and the accuracy are acceptable.
\end{rem}

\begin{rem}\label{rem:6}
Although the average secrecy unicast throughput $R_{S,1}^U$ is derived under high SNR, numerical results show that (\ref{Q3:14}) is also accurate for low and moderate SNR.
\end{rem}

\begin{rem}\label{rem:7}
For the case with perfect CSI, namely $\sigma^2_\zeta = 0$, unfortunately the exact closed-form expression for $R_{S,1}^U$ can not be achieved because the integration over lower incomplete gamma function $\gamma(\alpha,x)$ is difficult to obtain.
\end{rem}

\subsection{Performance analysis for SOS-based CSI}\label{subsec:sos}
In this subsection, the multicast outage and secrecy unicast throughput performance for the SOS-based CSI are analyzed in detail.

Here the BS only knows the distance information. Without loss of generality, we assume that $d_1 \le d_2 \le \cdots \le d_K$. The meaning of this case is that only the distance can be available in certain practical communications. As the channel gain is modeled as $\alpha_k = |g_k|^2 d_k^{-\eta}$, we can find that $\alpha_k,k \in \cal{K}$ are not necessarily ordered, that is $\alpha_j$ might be larger than $\alpha_k$ for $j>k$. However, the channel gain $\alpha_k$ is mainly dominated by the large-scale fading $d_k^{-\eta}$, which is also the motivation why we order the users according to their distances.

As to the optimization problem (\ref{1.9}), since the perfect channel gains $\alpha_k, k \in \cal{K}$ are not available at the BS, BS should set $j = 1, \theta_U^* = [\frac{\alpha_{K}- \frac{\epsilon_M}{\rho}}{\alpha_{K}(1+\epsilon_M)}]^+, \theta_M^* =1-\theta_U^*$.

Next, the multicast outage probability $P_{out}^{2} = Pr(\alpha_k < \frac{\epsilon_M}{\rho},\exists k \in \cal{K})$ and statistical expectation of secrecy unicast throughput $R_{S,2}^U = E[R_S^{U}]$ will be analyzed as follows.

\begin{thm}\label{thm:4}
As to $P_{out}^{2}$, the exact closed-form expression can be achieved as
\begin{equation}\label{Q2:1}\small
\begin{aligned}
&P_{out}^2 = 1- \prod_{k=1}^{K} \cr
&\left[2k \tbinom{K}{k} \sum \limits_{j=0}^{K-k} \tbinom{K-k}{j} \frac{(-1)^j}{D^{2(k+j)}} \frac{({\frac{\epsilon_M}{\rho})}^{-\frac{2(k+j)}{\eta}}}{\eta} \gamma{(\frac{2(k+j)}{\eta},\frac{\epsilon_M D^\eta}{\rho})}\right].
\end{aligned}
\end{equation}
Note that $\gamma(a,b)$ is a lower incomplete gamma function.
\end{thm}
\begin{proof}\label{prof:4}
See Appendix \ref{appe:4}.
\end{proof}

\begin{rem}\label{rem:8}
It is easy to obtain that the complexity of $P_{out}^2$ is $O(\frac{K(K+1)}{2})$. Compared to (\ref{pro:eq:1}), the complexity of achieving the multicast outage probability in SOS-based CSI is larger than that of perfect CSI, and it is also larger than that of imperfect CSI (shown in (\ref{Q3:a5})) in our simulation settings.
\end{rem}

\begin{rem}\label{rem:9}
The multicast outage expression (\ref{Q2:1}) is accurate for whole SNR, which is similar to that in imperfect CSI.
\end{rem}

\begin{rem}\label{rem:10}
Under high SNR, namely $\rho$ is large enough, according to $8.354.1$ in\cite{Zwillinger}, $(P_{out}^2)^\infty$ can be expressed as
\begin{equation}
(P_{out}^2)^\infty \approx 1- Z(\eta),
\end{equation}
where $Z(\eta)$ is a positive constant related to $\eta$. Therefore, we can conclude that there is no diversity gain for SOS-based CSI, which is similar to that in imperfect CSI.
\end{rem}

As to the average secrecy unicast throughput $R_{S,2}^U$, due to that BS only knows the distance information, we can not guarantee that channel gains are sorted according to their distances. For example, there are two users $U_1$ and $U_2$ in the cell, with $d_1\le d_2$. It is possible that $\alpha_1 < \alpha_2$. Note that the achieved security throughput depends on the order of joint channel gain. Hence, in general, when $K>2$, it is difficult to obtain a closed-form expression for $R_{S,2}^U$. In this paper, we only focus on the special case of $K=2$.

\begin{thm}\label{thm:5}
When $K =2$, $R_{S,2}^U$ can be approximated as
\begin{equation}\label{Q2:21}\small
\begin{aligned}
R_{S,2}^U &=\frac{4\pi}{l\eta D^4\ln2}\sum_{i=1}^l |\sin(\frac{2i-1}{2l}\pi)| \kappa_i^{\frac{2}{\eta}-1} \cr
&\{ J(i) \ln \nu -\frac{ \ln(\nu+\epsilon_M)\gamma(\frac{4}{\eta},\frac{(\kappa_i+1)\epsilon_M D^\eta}{\rho}) }{\eta(\kappa_i+1)[\frac{(\kappa_i+1)\epsilon_M}{\rho}]^{4/\eta}} \cr &+\frac{D\pi}{2q}\sum_{j=1}^q |\sin(\frac{2j-1}{2q}\pi)| x_j \bar{J}(i,x_j) \},
\end{aligned}
\end{equation}
where
\begin{equation*}\small
\begin{aligned}
&J(i)= \frac{D^4}{2} e^{-\frac{\epsilon_M D^\eta}{\rho}\kappa_i}-\frac{\gamma(\frac{4}{\eta},\frac{\epsilon_M \kappa_i D^\eta}{\rho})}{\eta[\frac{\epsilon_M \kappa_i}{ \rho}]^{4/\eta}}+\frac{D^4}{4(\kappa_i+1)}\cr
\end{aligned}
\end{equation*}
\begin{equation*}\footnotesize
\begin{aligned}
&\bar{J}(i,x) = -e^{\frac{\nu x^\eta}{\rho}}\text{Ei}(-\frac{\nu x^\eta}{\rho})(D^2 e^{-\frac{\epsilon_M D^\eta}{\rho}\kappa_i}- x^2 e^{-\frac{\epsilon_M x^\eta}{\rho}\kappa_i})\cr
&-\frac{x^2}{\kappa_i+1}e^{\frac{\nu(\kappa_i+1)x^\eta}{\rho}} [\text{Ei}(-\frac{\nu(\kappa_i+1)x^\eta}{\rho})-\text{Ei}(-\frac{(\nu+\epsilon_M)(\kappa_i+1)x^\eta}{\rho})]\cr
&-D^2[\frac{x^\eta \ln \nu}{D^\eta\kappa_i+x^\eta} - \frac{x^\eta}{D^\eta\kappa_i+x^\eta}e^{\frac{\nu(D^\eta\kappa_i+x^\eta)}{\rho}}\text{Ei}(-\frac{\nu(D^\eta\kappa_i+x^\eta)}{\rho})],
\end{aligned}
\end{equation*}\label{eq:thm:5}
\begin{equation}\small
\kappa_i = \frac{1}{2}(1+\cos(\frac{2i-1}{2l}\pi)),x_j = \frac{D}{2}(1+\cos(\frac{2j-1}{2q}\pi)).
\end{equation}
Note that $\nu = 1+ \epsilon_M$, $\text{Ei}(x) = \int_{-\infty}^x \frac{e^t}{t}\text{d}t, \text{for } x<0$, and $\gamma(a,b)$ is a lower incomplete gamma function. Here $q$ and $l$ are the number of terms included in the summation, which controls the approximation accuracy, due to the use of Gauss-Chebyshev integration\cite{Hildebrand}.
\end{thm}
\begin{proof}\label{prof:5}
See Appendix \ref{appe:5}.
\end{proof}

\begin{rem}\label{rem:11}
The computational complexity of $R_{S,2}^U$ can be easily obtained as $O(ql)$, which is smaller than that of $R_{S,1}^U$ in our simulation settings.
\end{rem}

\begin{rem}\label{rem:12}
Although the average secrecy unicast throughput $R_{S,2}^U$ is obtained under high SNR, simulation results show that (\ref{Q2:21}) is also accurate for low and moderate SNR.
\end{rem}

\subsection{Performance analysis in OMA systems}
In this subsection, the multicast outage and average secrecy unicast throughput achieved by the OMA scheme with partial CSI are correspondingly derived. Note that the OMA scheme is provided as a benchmark in our paper. As to the OMA scheme, a frame is equally divided into $2$ time slots, while the first one is used for the multi-cast traffic only, and the second one is singly for the unicast traffic. Then the achievable multicast rate in user $U_k$ can be expressed as
\begin{equation}\label{Q4:1}
R_k^M = \frac{1}{2} \log_2 (1+\alpha_k\rho),
\end{equation}
where $\frac{1}{2}$ denotes the time slot ratio for the multicast traffic. In the following, two partial CSI cases, namely imperfect CSI and SOS-based CSI, are carefully considered.

\subsubsection{Imperfect CSI}
The multicast outage of the OMA scheme $(P_{out}^1)^{OMA}$ can be denotes as
\begin{equation}\label{Q4:2}\small
(P_{out}^1)^{OMA} = 1 - Pr\{\hat{\alpha}_k \ge \frac{\lambda_M}{\rho}, \forall k \in {\cal{K}} \},
\end{equation}
where $\lambda_M = 2^{2R_M}-1$.
Therefore, $(P_{out}^1)^{OMA}$ can be easily approximated as
\begin{equation}\label{Q4:3}\small
\begin{aligned}
&(P_{out}^1)^{OMA} = 1 - [1-F_{\hat{\alpha}}(\frac{\lambda_M}{\rho})]^K \cr
&\mathop \approx \limits^{(a)} 1- [\frac{\pi}{cD} \sum \limits_{i=1}^{c} |\sin{\frac{2i-1}{2c}\pi}| x_i \exp{(-\frac{\lambda_M/\rho}{{x_i^{-\eta}}-\sigma_\zeta^2})}]^K,
\end{aligned}
\end{equation}
where derivation $(a)$ results from (\ref{Q3:a4}).

\begin{rem}\label{rem:13}
The computational complexity of $(P_{out}^1)^{OMA}$ is $O(c)$, which is same with that of $P_{out}^1$. As $\lambda_M>\epsilon_M$ and $x_i >0, i=1,\cdots,n$, then $(P_{out}^1)^{OMA}> P_{out}^1$, namely the NOMA scheme achieves a lower multicast outage probability (better outage preformance) than the OMA scheme.
\end{rem}

\begin{rem}\label{rem:14}
As to the perfect CSI, namely $\sigma_\zeta^2 = 0$, $(P_{out}^1)^{OMA}$ can be exactly obtained as
\begin{equation}\label{Q4:4}\small
(P_{out}^1)_{\sigma^2_\zeta = 0}^{OMA} = 1-[\frac{2}{\eta (\frac{\lambda_M}{\rho})^{\frac{2}{\eta}}D^2}\gamma(\frac{2}{\eta},\frac{\lambda_M D^\eta}{\rho})]^K.
\end{equation}
\end{rem}

\begin{rem}\label{rem:15}
Similarly, there is no diversity gain for both the imperfect CSI OMA and perfect CSI OMA schemes.
\end{rem}

As to the average security rate $(R_{S,1}^U)^{OMA}$, it can be easily expressed as
\begin{equation}\label{Q4:5}\small
(R_{S,1}^U)^{OMA} = \frac{1}{2}E\{\log_2(1+\hat{\alpha}_{\pi_1} \rho) - \log_2(1+\hat{\alpha}_{\pi_2} \rho)  \}.
\end{equation}
In the similar way of calculating $\Omega$ in (\ref{Q3:2}), $(R_{S,1}^U)^{OMA}$ can be approximated as
\begin{equation}\label{Q4:6}\small
\begin{aligned}
&(R_{S,1}^U)^{OMA} \approx  \frac{K \pi \rho}{4m\ln 2}  \sum_{u=1}^{m} |\sin \frac{2u-1}{2m}\pi|\cr
&\left\{\frac{1}{\rho\tau_u} - \sum_{r_0+r_1+ \cdots +r_n = K-1} A(r_1,\cdots,r_n) \bar{H}(r_1,\cdots,r_n,u) \right\}, \cr
\end{aligned}
\end{equation}
where
\begin{equation*}\label{Q4:611}\footnotesize
\begin{aligned}
&\bar{H}(r_1,\cdots,r_n,u)= -\frac{\pi}{nD\rho \tau_u} \sum \limits_{t=1}^{n} |\sin{\frac{2t-1}{2n}\pi}| x_t \bar{G}(r_1,\cdots,r_n,u,t)\cr
&+I_2(r_1,\cdots,r_n),\cr
\end{aligned}
\end{equation*}
\begin{equation}\label{Q3:17}\footnotesize
\begin{aligned}
\bar{G}(r_1,\cdots,r_n,u,t)&= 1-\frac{B(r_1,\cdots,r_n,u)}{\rho \tau_u \bar{\mu}_1}\cr
&+e^{\bar{\mu}_1}\text{Ei}(- \bar{\mu}_1)[\bar{\mu}_1-\frac{B(r_1,\cdots,r_n,u)}{\rho\tau_u} ].\cr
\end{aligned}
\end{equation}

\begin{rem}\label{rem:16}
The complexity of $(R_{S,1}^U)^{OMA}$ is $O\left(m\binom{K+n-1}{n}(n+n^2)\right)$, which is close to that of $R_{S,1}^U$. Meanwhile, $(R_{S,1}^U)^{OMA}$ is accurate for whole SNR range.
\end{rem}

\subsubsection{SOS-based CSI}
Similarly, The multicast outage of the OMA scheme $(P_{out}^2)^{OMA}$ can be exactly obtained as
\begin{equation}\label{Q4:21}\footnotesize
\begin{aligned}
&(P_{out}^2)^{OMA} = 1 - Pr\{\alpha_k \ge \frac{\lambda_M}{\rho}, \forall k \in {\cal{K}} \} \mathop = \limits^{(a)} 1 - \prod_{k=1}^{K} \cr
&\left[2k \tbinom{K}{k} \sum \limits_{j=0}^{K-k} \tbinom{K-k}{j} \frac{(-1)^j}{D^{2(k+j)}} \frac{({\frac{\lambda_M}{\rho})}^{-\frac{2(k+j)}{\eta}}}{\eta} \gamma{(\frac{2(k+j)}{\eta},\frac{\lambda_M D^\eta}{\rho})}\right], \cr
\end{aligned}
\end{equation}
where derivation $(a)$ results from (\ref{Q2:a3}).

\begin{rem}
The computational complexity of $(P_{out}^2)^{OMA}$ is $O(\frac{K(K+1)}{2})$, which is same with $P_{out}^2$. Note that $(P_{out}^2)^{OMA}$ is accurate for whole SNR range. Meanwhile, under high SNR condition, we can find that there is also no diversity gain.
\end{rem}

As to the average security rate $(R_{S,2}^U)^{OMA}$, we also consider the special case of $K=2$, that is
\begin{equation}\label{Q4:51}\fontsize{6pt}{5pt}
(R_{S,2}^U)^{OMA} = \frac{1}{2} \text{E}\left[\log_2 (1 + \rho \frac{|g_1|^2}{d_1^\eta}) -\log_2 (1 + \rho \frac{ |g_2|^2}{d_2^\eta}) | \frac{|g_1|^2}{d_1^\eta} \ge \frac{|g_2|^2}{d_2^\eta} \right].
\end{equation}

In the similar way of deriving $R_{S,2}^U$, $(R_{S,2}^U)^{OMA}$ can be calculated as
\begin{equation}\label{Q4:61}\footnotesize
\begin{aligned}
&(R_{S,2}^U)^{OMA} =\frac{\pi}{q D^3\ln2}\sum_{j=1}^q |\sin(\frac{2j-1}{2q}\pi)| x_j \bar{J}_{o,1}(x_j)+ \frac{\pi^2}{ql\eta D^3\ln2}\cr
&\sum_{i=1}^l |\sin(\frac{2i-1}{2l}\pi)| \kappa_i^{\frac{2}{\eta}-1}\left\{\sum_{j=1}^q |\sin(\frac{2j-1}{2q}\pi)| x_j \bar{J}_{o,2}(i,x_j)\right\},
\end{aligned}
\end{equation}
where
\begin{equation}\footnotesize
\begin{aligned}
\bar{J}_{o,1}(x)& =-e^{\frac{x^\eta}{\rho}}\text{Ei}(-\frac{x^\eta}{\rho})(D^2 - x^2 )\cr
\bar{J}_{o,2}(i,x) &=\frac{D^2 x^\eta}{D^\eta\kappa_i+x^\eta}e^{\frac{(D^\eta\kappa_i+x^\eta)}{\rho}}\text{Ei}(-\frac{(D^\eta\kappa_i+x^\eta)}{\rho})].
\end{aligned}
\end{equation}

\begin{rem}
The computational complexity of $(R_{S,2}^U)^{OMA}$ is $O(q+ql)$, which is close to that of $R_{S,2}^U$. Meanwhile, $(R_{S,2}^U)^{OMA}$ is accurate for the whole SNR range.
\end{rem}

\section{Numerical Results}\label{sec:rest}
In this section, numerical results are presented to validate the analytical expressions derived in this paper. In particular, the simulation parameters are set as follows. The disk radius $D$ is $5 m$, the path loss factor $\eta$ is $2$, and the number of users $K$ is $8$. Meanwhile, as to the number of included terms in the Gaussian-Chebyshev integration approximation, we set $c= 50, m = 5, n = 10, l = 100, q = 10$. The channel gain model has been discussed in subsection \ref{sub:sysmod}, and the channel estimation error $\sigma^2_\zeta$ in imperfect CSI is valued as $0.01$ in general if it is not specially pointed out. Note that all the numerical results are acquired by averaging over $1e5$ random channel gains realizations. Finally, in the following figures, ``sim" and ``cal" denote simulation results and derived analytical expression values, respectively.

\begin{figure}
  \centering
  \caption{Performance analysis of the imperfect CSI case vs the SNR $\rho$.}
  \subfigure[Outage performance $P_{out}^1$]{
  \includegraphics[width=0.5\textwidth]{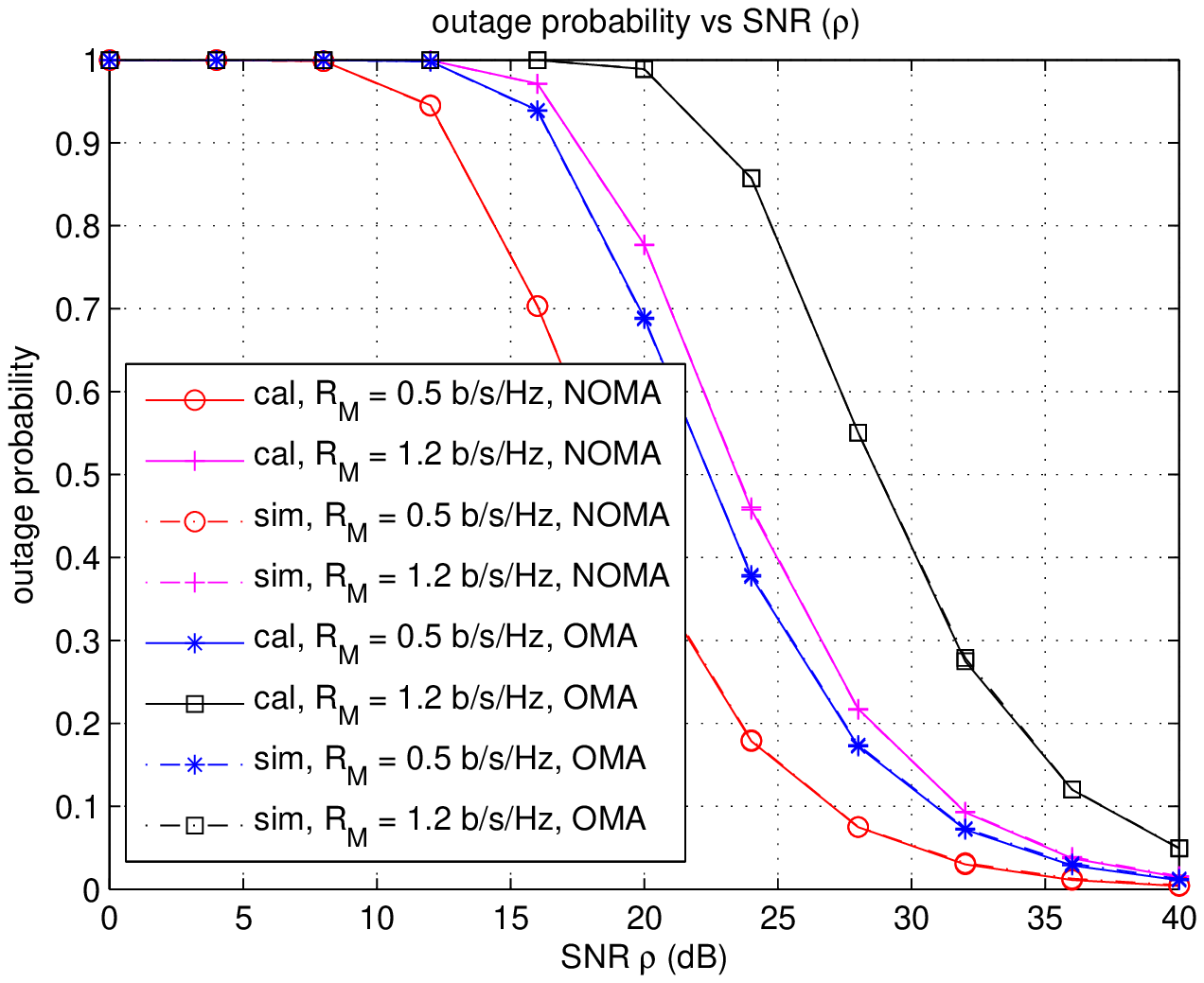}
  }
  \subfigure[Average secrecy unicast throughput performance $R_{S,1}^U$]{
  \includegraphics[width=0.5\textwidth]{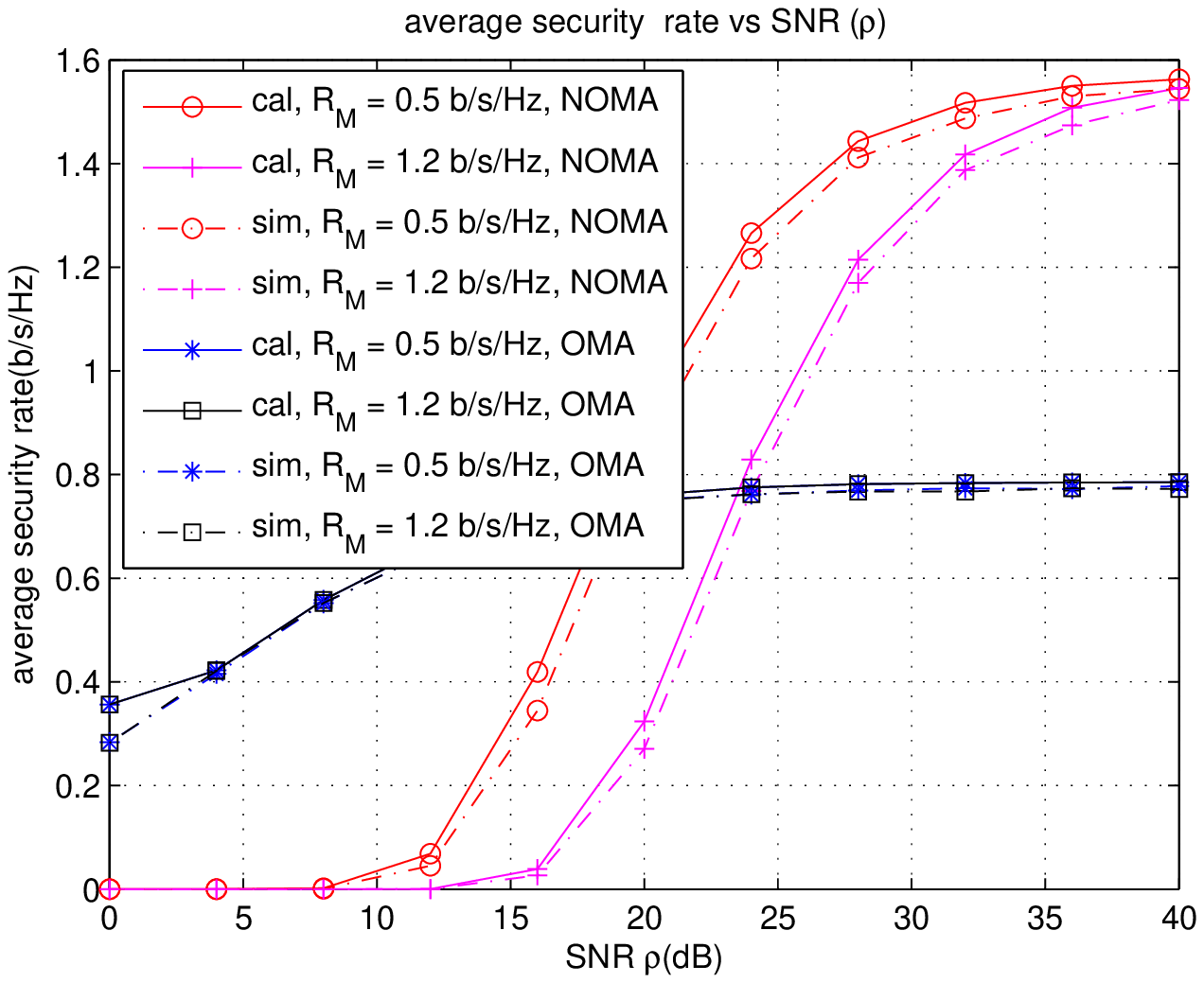}
  }
  \label{res:1}
\end{figure}

Fig. \ref{res:1} investigates the analytical performance with imperfect CSI as a function of the SNR $\rho$, where $R_M = 0.5 b/s/Hz$ and $R_M = 1.2 b/s/Hz$ are considered. On the one hand, the multicast outage probability $P_{out}^1$ is shown in $(a)$, and we can find that the derived expressions, namely (\ref{Q3:a5}) for the NOMA scheme and (\ref{Q4:3}) for the OMA scheme, match exactly with the Monte Carlo simulations. Furthermore, as $\rho$ increases, the multicast outage probability correspondingly decreases, that is the outage performance improves, which accords with common sense. In addition, a large $R_M$ will lead to a large $P_{out}^1$ as the QoS of the multicast traffic becomes strict. Finally, the NOMA scheme is always better than the OMA scheme as the multicast message can be sent in entire wireless resource in our NOMA design. On the other hand, $(b)$ studies the average secrecy unicast throughput $R_{S,1}^U$, and we can see that the analytical approximations ((\ref{Q3:14}) for the NOMA scheme and (\ref{Q4:6}) for the OMA scheme) also match well with the simulations. As discussed in \emph{Remark \ref{rem:5}}, due to the use of Gaussian-Chebyshev integration approximation, the important parameter $n$ should be carefully considered. The analytical accuracy increases as $n$ becomes large, while the computational complexity dramatically increases. Several values of $n$ have been tested and we find that $n=10$ can achieve acceptable accuracy and complexity. Meanwhile, It is easy to obtain that $R_{S,1}^U$ is proportional to $\rho$. Note that $R_{S,1}^U$ can be denoted as a unicast rate gap between the best channel gain user and the second best user. Therefore, the increment of $R_{S,1}^U$ decreases as $\rho$ increases because $\rho$ dominates the rate and the channel gain is less important. In particular, as to the OMA scheme, $(R_{S,1}^U)^{OMA}$ almost remains the same when $\rho \ge 20 \text{dB}$. In addition, $(R_{S,1}^U)^{OMA}$ is independent of $R_M$, which is in line with our OMA scheme, that is the unicast message is singly sent in part time slot. Finally, a important observation is that the NOMA scheme shows great advantages over the OMA scheme in good condition (large $\rho$), but is inferior to the OMA scheme in bad condition (small $\rho$). The reason lies on that under low SNR, the interference from the multicast message is dominant, while the spectral efficiency is in the ascendant under high SNR, thus the characteristics of both the NOMA and OMA schemes are fully validated.

\begin{figure}
  \centering
  \caption{Performance analysis of the imperfect CSI case vs the channel estimation error $\sigma^2_\zeta$.}
  \subfigure[Outage performance $P_{out}^1$]{
  \includegraphics[width=0.5\textwidth]{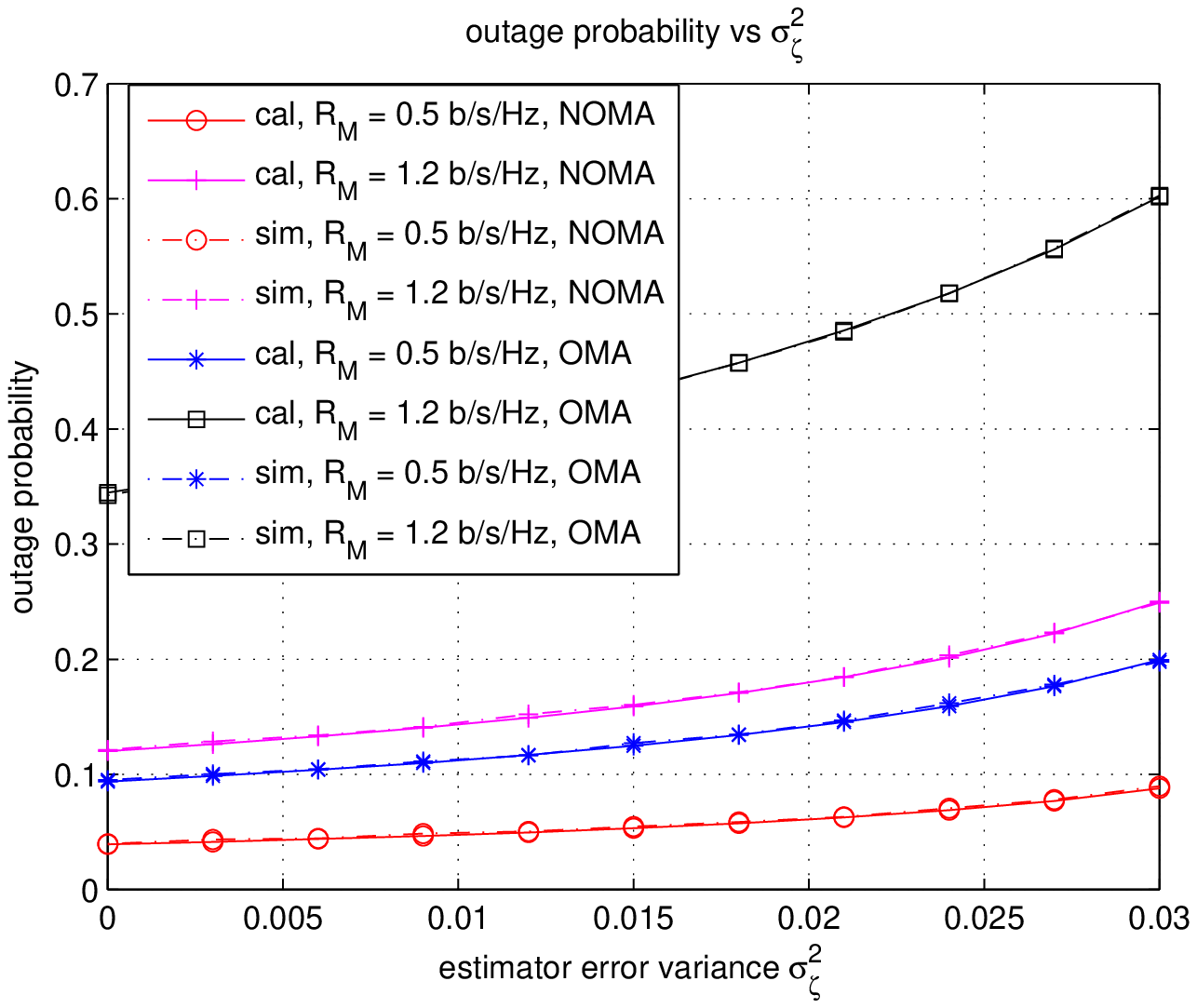}
  }
  \subfigure[Average secrecy unicast throughput performance $R_{S,1}^U$]{
  \includegraphics[width=0.5\textwidth]{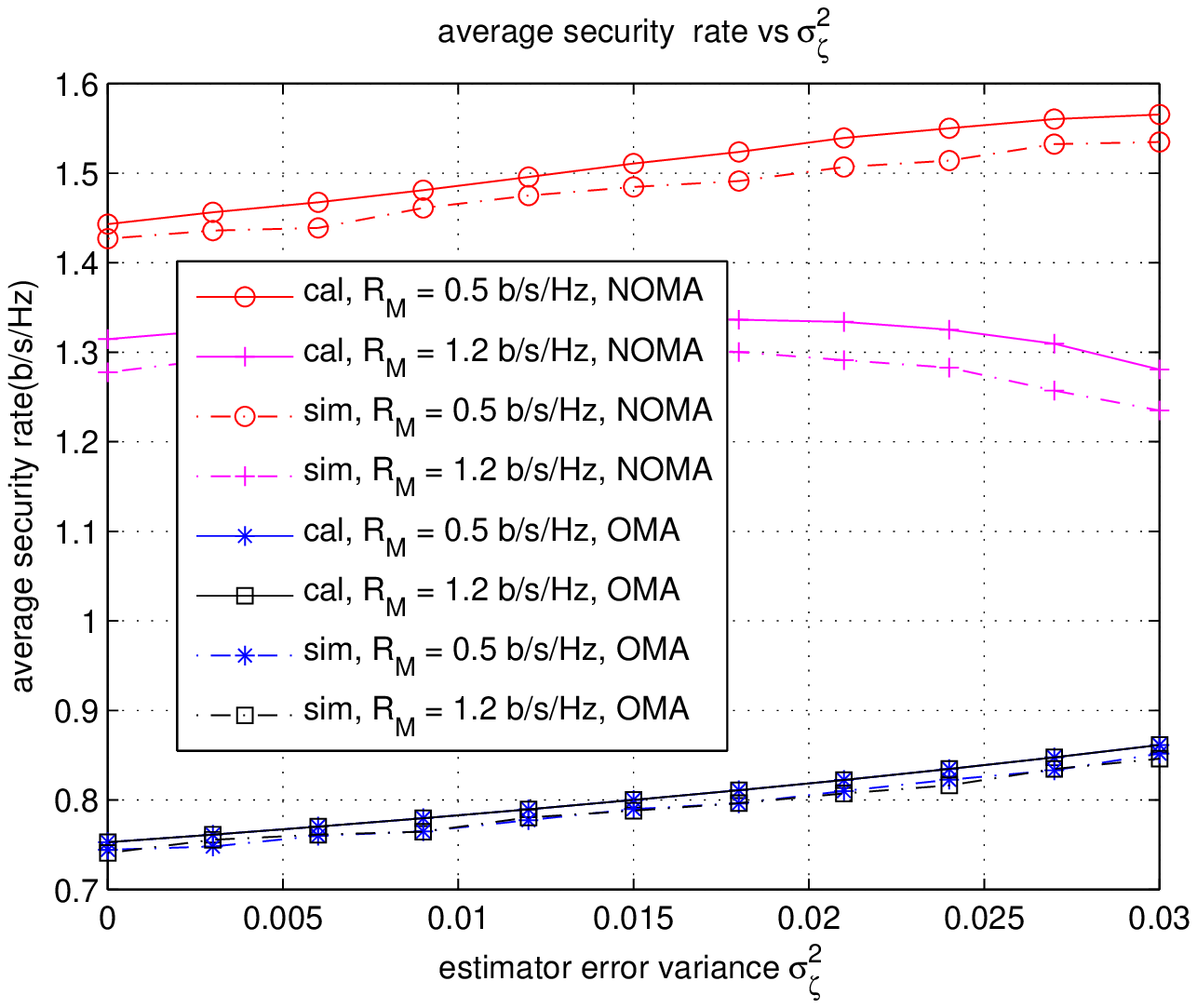}
  }
  \label{res:2}
\end{figure}

Fig. \ref{res:2} illustrates the analytical performance with imperfect CSI as a function of the channel estimation error $\sigma^2_\zeta$, where $\rho = 30 dB$. It is easy to find that the analytical expressions for both the outage probability and average secrecy unicast throughput match well with the simulations. As to the $P_{out}^1$ described in $(a)$, it is proportional to $\sigma^2_\zeta$, namely the multicast outage performance deteriorates when the channel estimation error increases, since higher estimation error brings stronger interference. What is more, the NOMA scheme always achieves better outage performance compared to the OMA scheme. On the other hand, a surprising observation of $R_{S,1}^U$ shown in $(b)$ is that it does not always decrease as the estimation error becomes large. Note that the average secrecy unicast throughput is a unicast rate gap between the best channel gain user and the second best user. Therefore, for the interference-insensitive case, such as the OMA scheme and a good condition in the NOMA scheme, $R_{S,1}^U$ increases with $\sigma^2_\zeta$ as the rate of second best user greatly decreases, while for the interference-sensitive case, such as a bad condition in the NOMA scheme, $R_{S,1}^U$ decreases with $\sigma^2_\zeta$, which results from the dominant decrement in the rate of the best user. Furthermore, as the current settings are a good condition, namely a small $\sigma^2_\zeta$ and a large $\rho$, the NOMA scheme obtains a larger secrecy unicast throughput compared to the OMA scheme. Finally, as to the NOMA scheme, a large $R_M$ will lead to a small $R_{S,1}^U$, while $(R_{S,1}^U)^{OMA}$ remains the same regardless of $R_M$, which can be easily explained according to our system design.

\begin{figure}
  \centering
  \caption{Performance analysis of the imperfect CSI case vs the number of users $K$.}
  \subfigure[Outage performance $P_{out}^1$]{
  \includegraphics[width=0.5\textwidth]{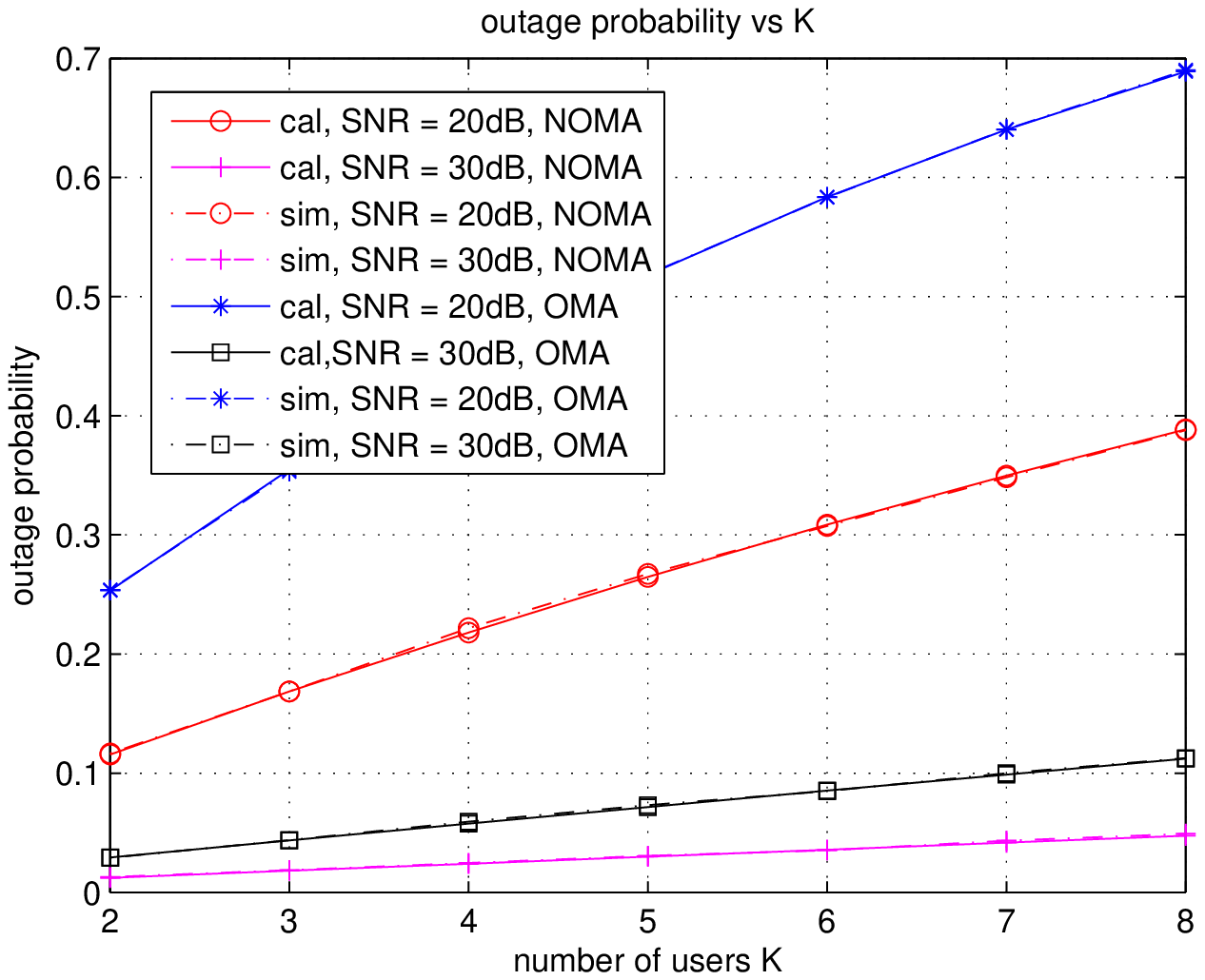}
  }
  \subfigure[Average secrecy unicast throughput performance $R_{S,1}^U$]{
  \includegraphics[width=0.5\textwidth]{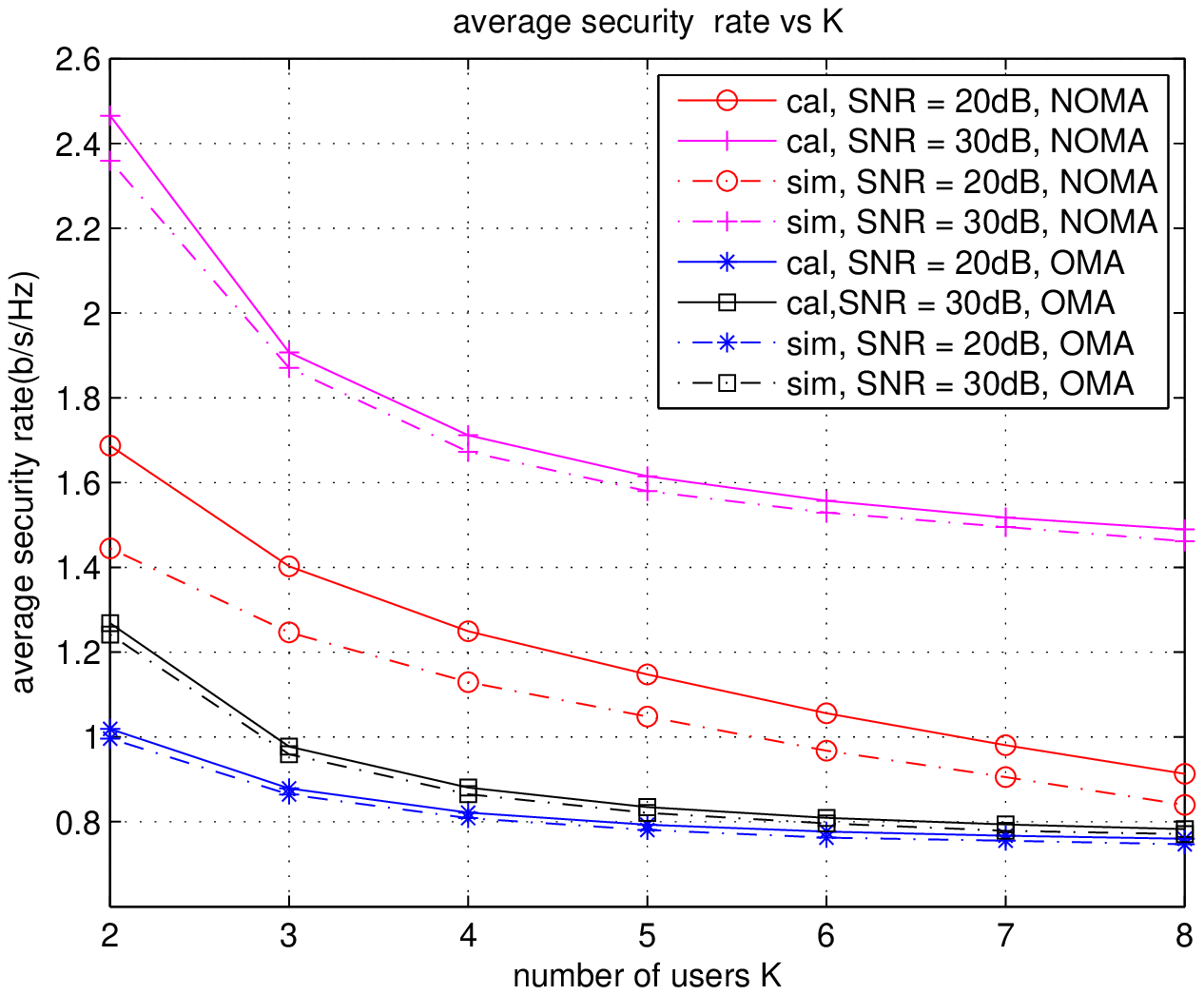}
  }
  \label{res:3}
\end{figure}

Fig. \ref{res:3} studies the analytical performance with imperfect CSI as a function of the number of users $K$, where $R_M = 0.5 b/s/Hz$. Note that both the multicast outage and secrecy unicast throughput performance deteriorates as $K$ increases, which results from the more strict QoS of the multicast traffic when the number of users becomes larger. Meanwhile, no matter the multicast outage or the secrecy unicast throughput performance, the NOMA scheme shows great advantages over the OMA scheme as the current parameter settings are a good condition. Finally, the performance improves as $\rho$ increases, which accords with common sense.

\begin{figure}
  \centering
  \caption{Performance analysis of the SOS-based CSI case vs the SNR $\rho$.}
  \subfigure[Outage performance $P_{out}^2$]{
  \includegraphics[width=0.5\textwidth]{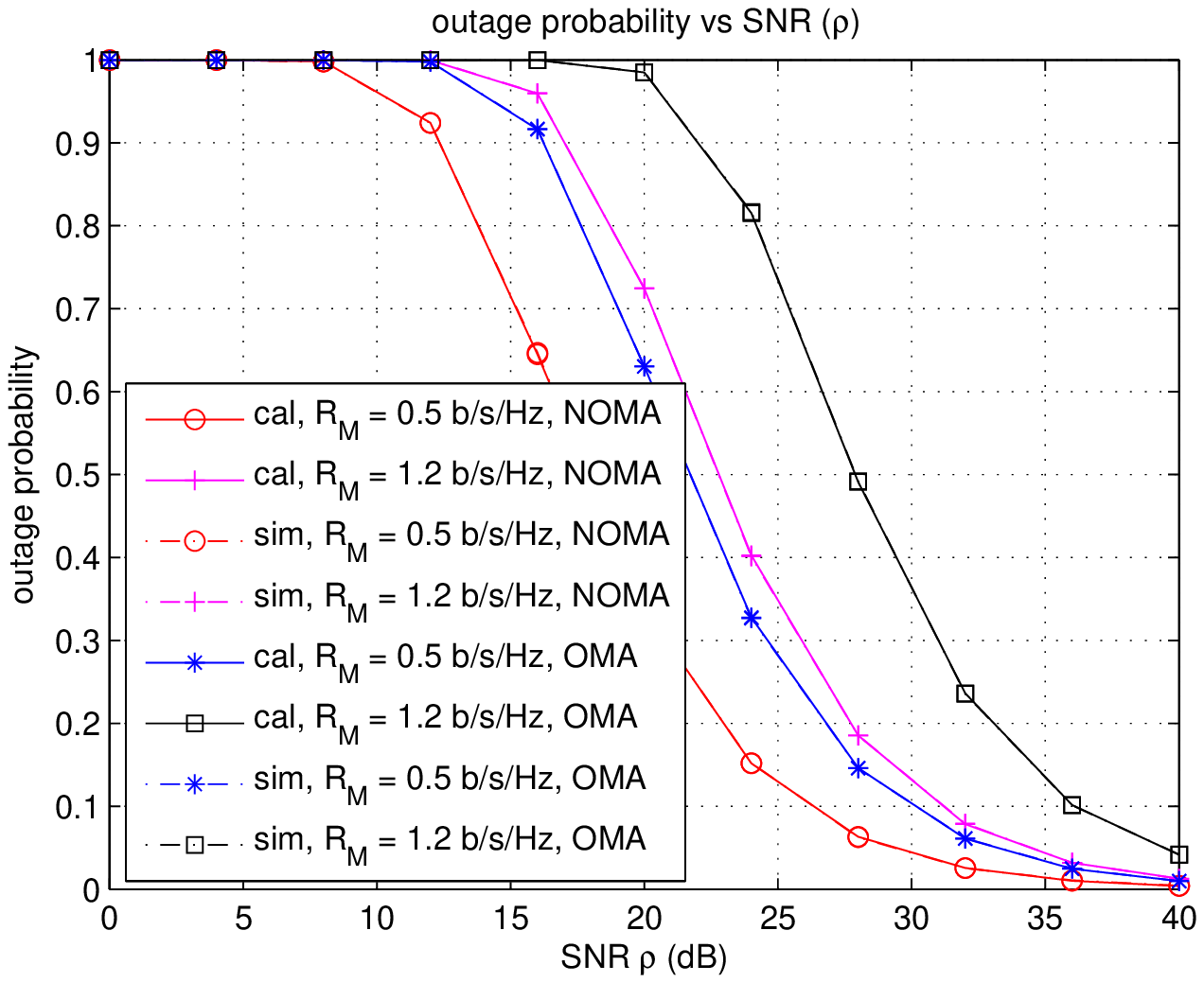}
  }
  \subfigure[Average secrecy unicast throughput performance $R_{S,2}^U$]{
  \includegraphics[width=0.5\textwidth]{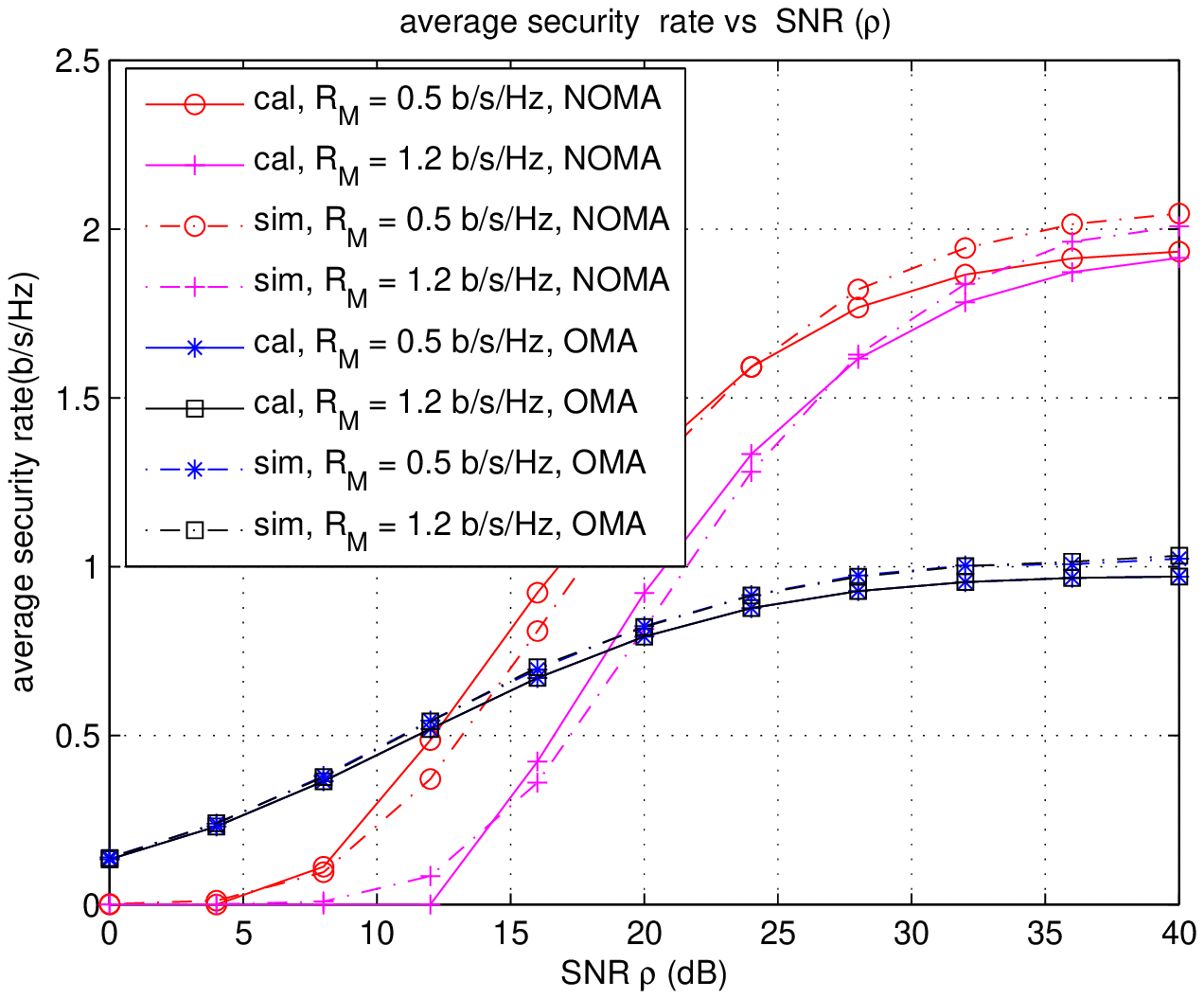}
  }
  \label{res:4}
\end{figure}

Fig. \ref{res:4} shows the analytical performance with SOS-based CSI as a function of the SNR $\rho$, where $R_M = 0.5 b/s/Hz$ and $R_M = 1.2 b/s/Hz$ are considered. Note that for the multicast outage analysis, $K = 8$ is taken into account, while for the average secrecy unicast throughput, we accordingly set $K = 2$. As to the multicast outage probability $P_{out}^2$, it can see that the derived expressions, (\ref{Q2:1}) for the NOMA scheme and (\ref{Q4:21}) for the OMA scheme, match exactly with the Monte Carlo simulations. Furthermore, the outage performance improves as $\rho$ increases, which is in line with common sense. In addition, $P_{out}^2$ is proportional to $R_M$ due to the more strict QoS of the multicast traffic as $R_M$ becomes larger. Finally, the NOMA scheme achieves better outage performance compared to the OMA scheme as the multicast message can only be sent in part of wireless resource in the OMA scheme. On the other hand, as to $R_{S,2}^U$, we can conclude that the analytical approximations, namely (\ref{Q2:21}) for the NOMA scheme and (\ref{Q4:61}) for the OMA scheme, also match well with the simulations. Meanwhile, $R_{S,2}^U$ accordingly increases with $\rho$ and the increment of $R_{S,2}^U$ decreases as $\rho$ becomes large, which has been discussed in Fig. \ref{res:1}. In particular, as to the OMA scheme, $(R_{S,2}^U)^{OMA}$ almost remains the same when $\rho \ge 30 \text{dB}$. In addition, $(R_{S,2}^U)^{OMA}$ is independent of $R_M$, which results from the fact that the unicast message is singly sent in part time slot in the OMA scheme. Finally, an outstanding observation is that the NOMA scheme achieves a larger $R_{S,2}^U$ than the OMA scheme in good condition (large $\rho$), but obtains a smaller $R_{S,2}^U$ compared to the OMA scheme in bad condition (small $\rho$), which has been analyzed in Fig. \ref{res:1} and fully demonstrates the characteristics of the NOMA and OMA schemes.

\begin{figure}
  \centering
  \caption{Performance comparison between imperfect CSI and SOS-based CSI.}
  \subfigure[Outage performance]{
  \includegraphics[width=0.5\textwidth]{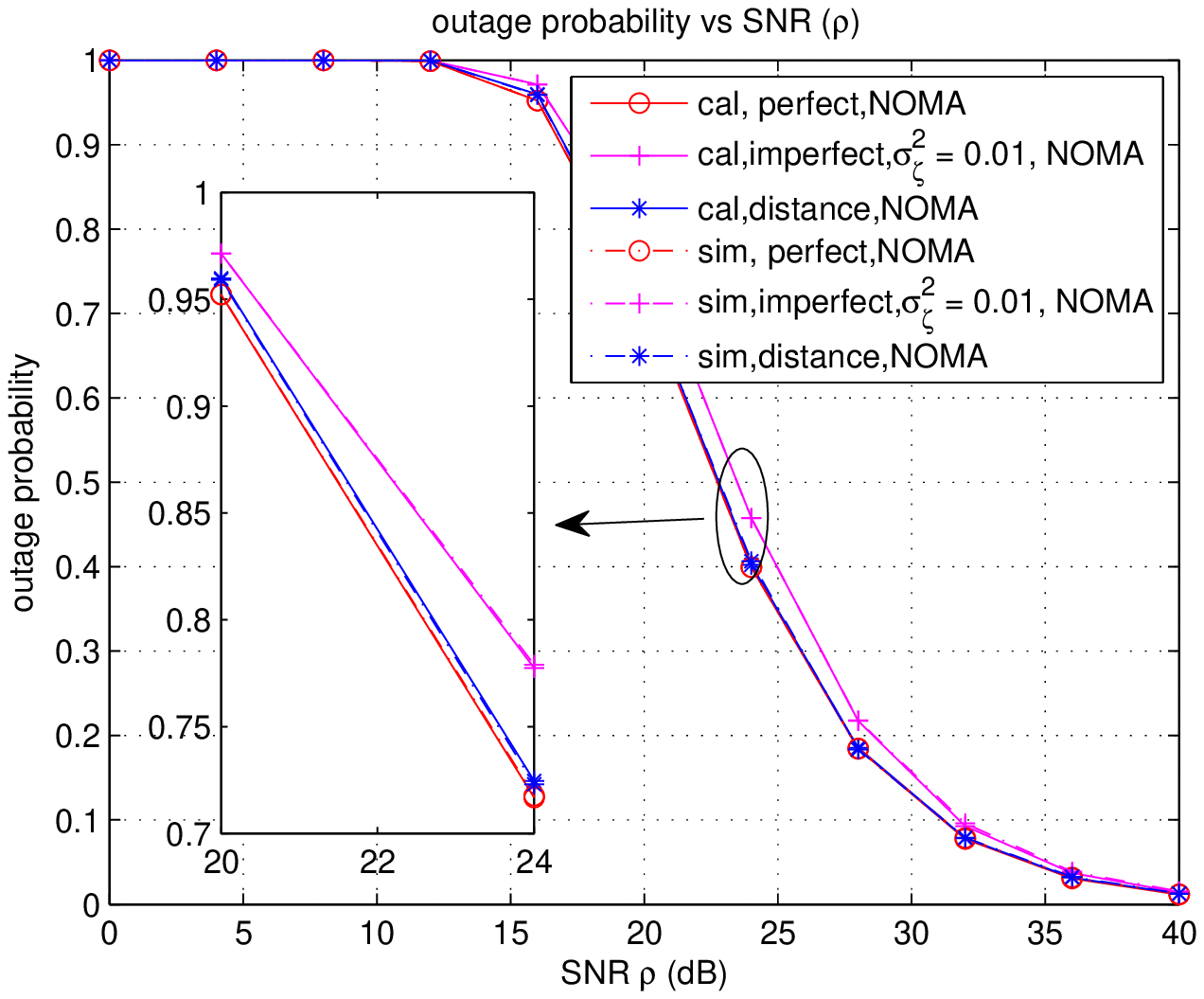}
  }
  \subfigure[Average secrecy unicast throughput performance]{
  \includegraphics[width=0.5\textwidth]{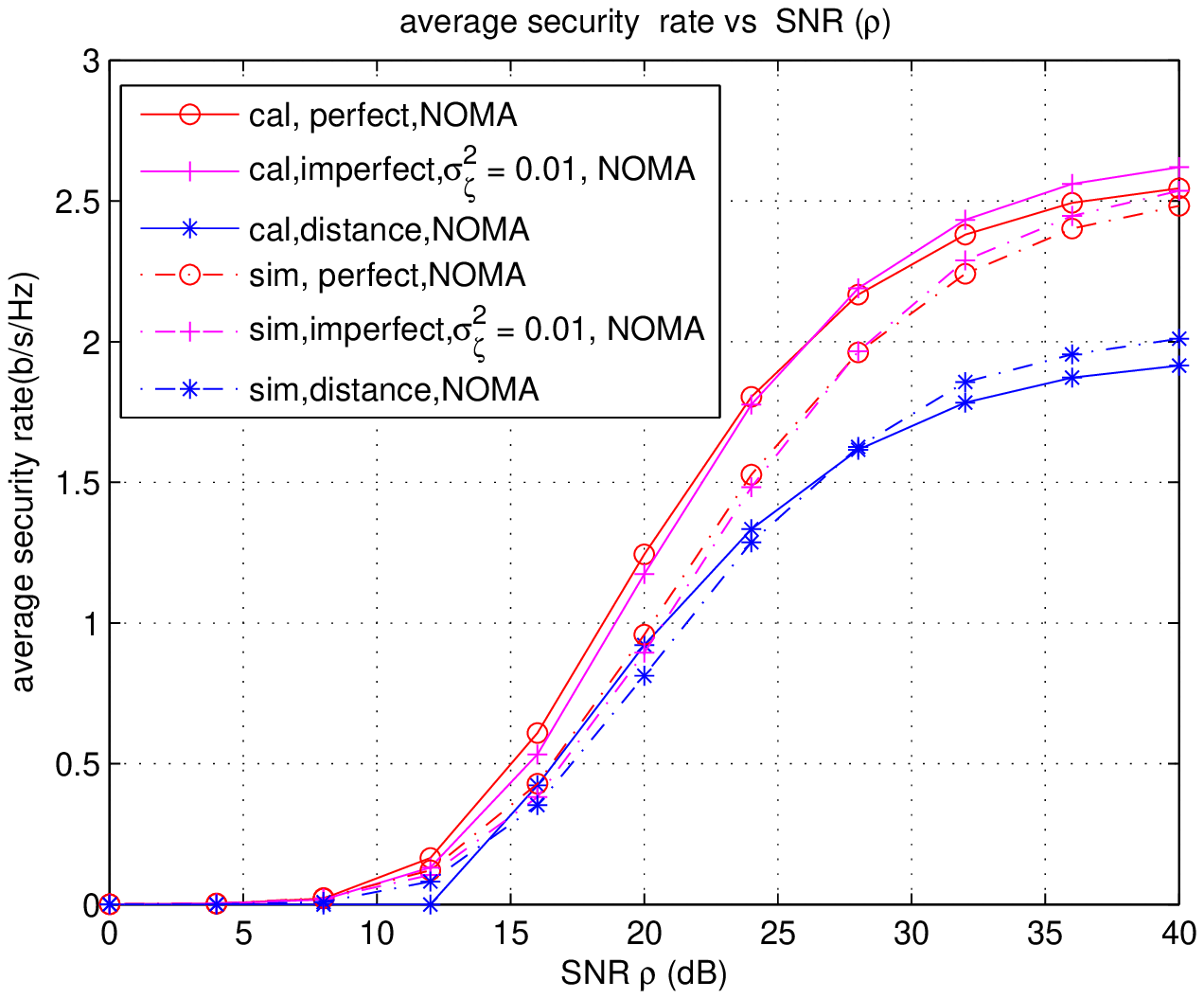}
  }
  \label{res:5}
\end{figure}

In Fig. \ref{res:5}, a performance comparison between imperfect CSI and SOS-based CSI is made, where $R_M = 1.2 b/s/Hz$. Similarly, $K$ is $8$ for the multicast outage performance, while for the average secrecy unicast throughput, we set $K =2$. As to the multicast outage performance, we can see that the perfect CSI NOMA scheme slightly outperforms the SOS-based CSI NOMA scheme, and the SOS-based CSI NOMA scheme is slimly better than the imperfect CSI NOMA scheme with $\sigma^2_\zeta = 0.01$. However, the difference is almost negligible. On the other hand, as to the average secrecy unicast throughput, the imperfect CSI NOMA scheme (no matter $\sigma^2_\zeta = 0$ or $\sigma^2_\zeta = 0.01$) shows great advantages over the SOS-based CSI NOMA scheme under high SNR, and achieves almost the same under low SNR. In addition, it is worth pointing out that the approximate expressions are close to the Monte Carlo simulations.

\section{Conclusion}\label{sec:conl}
In this paper, we have studied the multicast outage and secrecy unicast throughput performance in a downlink single-cell NOMA network with partial CSI, while a mixed multicast and unicast traffic scenario is considered. In particular, two types of partial CSI, namely the imperfect CSI and SOS-based CSI, are taken into account. For the two NOMA schemes, the closed-form approximations of the performance are derived except that the approximate expression of the secrecy unicast throughput in SOS-based CSI only regards two users. The provided numerical results confirm that the derived approximate expressions of both two cases for the multicast outage probability and secrecy unicast throughput match well with the Monte Carlo simulations. Meanwhile, simulation results demonstrate that the NOMA scheme considered in both two cases achieves better outage performance compared to the OMA scheme. However, as to the secrecy unicast throughput, the NOMA scheme shows great advantages over the OMA scheme in good condition (high SNR), but is inferior to the OMA scheme in bad condition (low SNR). Finally, the two NOMA schemes achieve similar performance except that the NOMA scheme with imperfect CSI obtains larger secrecy unicast throughput than that based on SOS under high SNR.

\appendices
%%%%%%%%%%%%%%%%%%%%%%%%%%%%%%%%%%%%%%%%%%%%%%%%%%%%%%
\section{Proof of Theorem \ref{thm:1}}\label{appe:1}
First, we prove that the optimal solution can be obtained when the constraint (\ref{2.1}) achieves equality, which can be demonstrated by the contradiction method. Here we assume that the optimal solution is $(\theta_M^*.\theta_U^*)$, where $R_{\pi_K}^M (\theta_M^*,\theta_U^*)> R_M$, namely $\frac{\theta_M^* \alpha_{\pi_K}}{\theta_U^* \alpha_{\pi_K}+\frac{1}{\rho}} > 2^{R_M}-1$. Furthermore, constraint condition (\ref{1.2}) is also guaranteed, that is $\theta_M^* + \theta_U^* \le 1$. It is easy to find that $R_{\pi_K}^M(\theta_M,\theta_U)$ increases with $\theta_M$ and decreases with $\theta_U$. Therefore, we can construct a new solution $(\bar{\theta}_M,\bar{\theta}_U)$ as $\bar{\theta}_M = \theta_M^* - \Delta \theta,\bar{\theta}_U= \theta_U^* + \Delta \theta$, where $\Delta \theta>0$ and $R_{\pi_K}^M(\bar{\theta}_M,\bar{\theta}_U) = R_M$. Note that $\bar{\theta}_M + \bar{\theta}_U = \theta_M^* +\theta_U^* \le 1$, namely constraint (\ref{1.2}) is also satisfied with regard to $(\bar{\theta}_M,\bar{\theta}_U)$. As optimization objective $R_S^U(\theta_U)$ singly increases with $\theta_U$, we can easily conclude that $R_S^U(\bar{\theta}_M,\bar{\theta}_U) > R_S^U(\theta_M^*,\theta_U^*)$, which contradicts our assumption that $(\theta_M^*.\theta_U^*)$ is the optimal solution. Hence, constraint (\ref{2.1}) should achieve equality in order to obtain the optimal solution of problem (\ref{eq:a1}).

Second, as to the optimal solution of problem (\ref{eq:a1}), constraint (\ref{1.2}) should also set equality. In the similar way, we assume that the optimal solution is $(\theta_M^*.\theta_U^*)$, where $\theta_M^* + \theta_U^* < 1$ and $R_{\pi_K}^M (\theta_M^*,\theta_U^*)= R_M$. Correspondingly, we construct a new solution $(\tilde{\theta}_M,\tilde{\theta}_U)$ as $\tilde{\theta}_M = \theta_M^* + \Delta{\theta}_M,\tilde{\theta}_U = \theta_U^* + \Delta{\theta}_U$, where $\Delta{\theta}_M >0, \Delta{\theta_U}>0$, and $R_{\pi_K}^M (\tilde{\theta}_M,\tilde{\theta}_U)= R_M, \tilde{\theta}_M + \tilde{\theta}_U = 1$. Hence $(\Delta \theta_M, \Delta \theta_U)$ need to satisfy
\begin{equation}\label{eq:app1}
\begin{aligned}
&\Delta \theta_M+ \Delta \theta_U = 1- \theta_M^* -\theta_U^* \cr
&\Delta \theta_M  = (2^{R_M}-1)  \Delta \theta_U.
\end{aligned}
\end{equation}
Then we can obtain that $\Delta \theta_M = \frac{(2^{R_M}-1)(1- \theta_M^* -\theta_U^*)}{2^{R_M}}, \Delta \theta_U = \frac{1- \theta_M^* -\theta_U^*}{2^{R_M}}$. As $R_S^U(\theta_U)$ singly increases with $\theta_U$, we can conclude that $R_S^U(\tilde{\theta}_M,\tilde{\theta}_U) > R_S^U(\theta_M^*,\theta_U^*)$, which contradicts our assumption that $(\theta_M^*.\theta_U^*)$ is the optimal solution. Therefore, in order to achieve the optimal solution of problem (\ref{eq:a1}), constraint (\ref{1.2}) should set equality.

In conclusion, the optimal solution of problem (\ref{eq:a1}) can be achieved while both (\ref{1.2}) and (\ref{2.1}) achieve equality, and \emph{Theorem \ref{thm:1}} is completely proved.

\section{Derivation of $P_{out}^1$}\label{appe:2}
In order to derive $P_{out}^1$, we first need to calculate the probability distribution function (PDF) and the cumulative distribution function (CDF) of the unordered channel gain $\hat{\alpha}_k$. As all the users are all uniformly deployed in the disc with radius $D$, it is easy to find that $f_{d_k}(x) = \frac{2x}{D^2},F_{d_k}(x) = \frac{x^2}{D^2},k\in{\cal{K}}$. Since $\alpha_k = \hat{\alpha}_k+\zeta$, then $\hat{\alpha}_k \sim {\cal{CN}}(0,d_k^{-\eta}-\sigma^2_{\zeta})$. Therefore, for the given $d_k$, the conditional probability and cumulative distribution of $\hat{\alpha}_k$ can be achieved as $f_{\hat{\alpha}_k|d_k} (y|d_k)= \frac{1}{d_k^{-\eta}-\sigma_\zeta^2} \exp(-\frac{y}{d_k^{-\eta}-\sigma_\zeta^2}), F_{\hat{\alpha}_k|d_k} (y|d_k) = 1- \exp(-\frac{y}{d_k^{-\eta}-\sigma_\zeta^2})$. Accordingly $f_{\hat{\alpha}_k}(y)$ and $F_{\hat{\alpha}_k}(y)$ can be obtained as
\begin{equation}\label{Q3:a1}\footnotesize
\begin{aligned}
&f_{\hat{\alpha}_k}(y) = f_{\hat{\alpha}}(y) = \int_0^D f_{\hat{\alpha}_k|d_k} (y|d_k) f_{d_k}(x) dx \cr
&= \frac{2}{D^2} \int_0^D \frac{x}{x^{-\eta}-\sigma_\zeta^2} \exp(-\frac{y} {x^{-\eta}-\sigma_\zeta^2})dx.\cr
&F_{\hat{\alpha}_k}(y) = F_{\hat{\alpha}}(y) =\int_0^D F_{\hat{\alpha}_k|d_k} (y|d_k) f_{d_k}(x) dx \cr
&= 1- \frac{2}{D^2} \int_0^D x \exp(-\frac{y}{{x^{-\eta}}-\sigma_\zeta^2}) dx.
\end{aligned}
\end{equation}

Second, based on the $2.1.6$ in \cite{David}, the PDF and CDF of the sorted channel gain $\hat{\alpha}_{\pi_i}$ can be expressed as
\begin{equation}\label{Q3:a3}\footnotesize
\begin{aligned}
&f_{\hat{\alpha}_{\pi_{i}}}(z) = i \tbinom{K}{i} (F_{\hat{\alpha}}(z))^{K-i} (1-F_{\hat{\alpha}} (z))^{i-1} f_{\hat{\alpha}}(z).\cr
F_{\hat{\alpha}_{\pi_i}}(z) &= \int_0^z f_{\hat{\alpha}_{\pi_i}}(x) dx \mathop = \limits^{(a)} i  \tbinom{K}{i} \int_0^{F_{\hat{\alpha}}(z)} \sum \limits_{s=0}^{i-1} \tbinom{i-1}{s} (-1)^s t^{K-i+s} dt  \cr
&= i  \tbinom{K}{i} \sum \limits_{s=0}^{i-1}(-1)^s  \tbinom{i-1}{s} \frac{(F_{\hat{\alpha}}(z))^{K-i+s+1}}{K-i+s+1},
\end{aligned}
\end{equation}
where derivation $(a)$ applies the binomial theorem to $ (1-F_{\hat{\alpha}} (z))^{i-1}$.

Third, we can see that $F_{\hat{\alpha}}(y)$ is difficult to obtain due to the calculation of $I_1(y) = \int_0^D x \exp(-\frac{y} {{x^{-\eta}}-\sigma_\zeta^2}) dx$. Therefore, by employing the Gauss-Chebyshev integration \cite{Hildebrand}, $I_1(y)$ can be approximated as
\begin{equation}\label{eq:app2}\footnotesize
I_1(y) \approx \frac{\pi D}{2c} \sum \limits_{i=1}^c |\sin{\frac{2i-1}{2c}\pi}|x_i \exp{(-\frac{y}{{x_i^{-\eta}}-\sigma_\zeta^2})},
\end{equation}
where $x_i = \frac{D}{2} (1+\cos{(\frac{2i-1}{2c}\pi)})$, and $c$ is the number of terms included in the summation, which controls the approximation accuracy. Thus, $F_{\hat{\alpha}}(y)$ can be accordingly approximated as
\begin{equation}\label{Q3:a4}\footnotesize
F_{\hat{\alpha}}(y) \approx 1-\frac{\pi}{cD} \sum \limits_{i=1}^{c} |\sin{\frac{2i-1}{2c}\pi}| x_i \exp{(-\frac{y}{{x_i^{-\eta}}-\sigma_\zeta^2})} .
\end{equation}

Based on (\ref{Q3:a3}) and (\ref{Q3:a4}), the multicast outage probability $P_{out}^1$ can be approximated as
\begin{equation}\footnotesize
\begin{aligned}
&P_{out}^1 = F_{\hat{\alpha}_{\pi_K}} (\frac{\epsilon_M}{\rho}) \approx K \sum \limits_{s=0}^{K-1} \frac{(-1)^s \tbinom{K-1}{s}}{s+1} \cr
&(1-\frac{\pi}{cD} \sum \limits_{i=1}^{c} |\sin{\frac{2i-1}{2c}\pi}| x_i \exp{(-\frac{\epsilon_M}{\rho (x_i^{-\eta}-\sigma_\zeta^2)})})^{s+1}\cr
& =  1- [\frac{\pi}{cD} \sum \limits_{i=1}^{c} |\sin{\frac{2i-1}{2c}\pi}| x_i  \exp(-\frac{\epsilon_M }{\rho(x_i^{-\eta}-\sigma_\zeta^2)})]^K.
\end{aligned}
\end{equation}

Therefore, \emph{Theorem \ref{thm:2}} is completely proved.

\section{Proof of \emph{Proposition \ref{pro:1}}}\label{appe:2.1}
When $\sigma^2_\zeta = 0$, according to $3.326.4$ in\cite{Zwillinger}, the $I_1(y)$ in (\ref{eq:app2}) can be exactly obtained as
\begin{equation}\label{eq:prop1}\footnotesize
I_1(y) = \frac{1}{\eta y^{\frac{2}{\eta}}}\gamma(\frac{2}{\eta},yD^\eta).
\end{equation}
Hence, $(P_{out}^1)_{\sigma^2_\zeta = 0}$ can be calculated as
\begin{equation}\label{eq:prop2}\footnotesize
\begin{aligned}
(P_{out}^1)_{\sigma^2_\zeta = 0}& = F_{\hat{\alpha}_{\pi_K}} (\frac{\epsilon_M}{\rho})= 1-(1-F_{\hat{\alpha}}(\frac{\epsilon_M}{\rho}))^K\cr
&=1-[\frac{2}{\eta {(\frac{\epsilon_M}{\rho})}^{\frac{2}{\eta}}D^2}\gamma(\frac{2}{\eta},\frac{\epsilon_M D^\eta}{\rho})]^K.
\end{aligned}
\end{equation}

\section{Derivation of $R_{S,1}^U$}\label{appe:3}
First, it is easy to see that $R_{U,out}^1 = 0$ if the multicast outage happens; otherwise, $R_{U,suc}^1 = \log_2{(1+\rho \hat{\alpha}_{\pi_1} \theta_U^*)}- \log_2{(1+\rho \hat{\alpha}_{\pi_2} \theta_U^*)}$, where $\theta_U^* = \frac{\hat{\alpha}_{\pi_K}- \frac{\epsilon_M}{\rho}}{\hat{\alpha}_{\pi_K}(1+\epsilon_M)}$. Therefore, $R_{S,1}^U$ can be expressed as
\begin{equation}\scriptsize
R_{S,1}^U = P_{out}^1 E\{R_{U,out}^1\} + (1-P_{out}^1) E\{R_{U,suc}^1\} = (1-P_{out}^1) E\{R_{U,suc}^1\}.
\end{equation}

Next, according to the $2.2.1$ in\cite{David}, the joint distribution of $\hat{\alpha}_{\pi_i}$ and $\hat{\alpha}_{\pi_j}$ ($i<j$) is
\begin{equation}\label{Q3:1}\scriptsize
\begin{aligned}
f_{\hat{\alpha}_{\pi_i},\hat{\alpha}_{\pi_j}} (x,y) &= \frac{K!}{(K-j)!(j-i-1)!(i-1)!} f_{\hat{\alpha}}(x) (1-F_{\hat{\alpha}}(x))^{i-1}\cr
& [F_{\hat{\alpha}}(x)-F_{\hat{\alpha}}(y)]^{j-i-1} (F_{\hat{\alpha}}(y))^{K-j} f_{\hat{\alpha}}(y).
\end{aligned}
\end{equation}

Here the exact closed-form expression of $R_{S,1}^U$ is difficult to obtain due to the complication of $\theta_U^*$. Then an approximate expression under high SNR is accordingly derived. Note that when $\rho$ is large enough, $\theta_U^* \approx \frac{1}{1+\epsilon_M}$. Therefore, under high SNR, $R_{S,1}^U$ can be approximately written as
\begin{equation}\label{Q3:2}\scriptsize
\begin{aligned}
R_{S,1}^U & \approx (1-P_{out}^1) E\{\log_2(1+\frac{\rho\hat{\alpha}_{\pi_1}}{1+\epsilon_M})-\log_2(1+\frac{\rho\hat{\alpha}_{\pi_2}}{1+\epsilon_M}) \} \cr
& = (1-P_{out}^1) \Omega,
\end{aligned}
\end{equation}
where $\Omega = E\{\log_2 (1+\epsilon_M + \rho\hat{\alpha}_{\pi_1}) - \log_2 (1+\epsilon_M+\rho \hat{\alpha}_{\pi_2}) \}$.

According to (\ref{Q3:1}), the joint distribution of $\hat{\alpha}_{\pi_1}$ and $\hat{\alpha}_{\pi_2}$ can be achieved as
\begin{equation}\label{Q3:5}\scriptsize
\begin{aligned}
f_{\hat{\alpha}_{\pi_1},\hat{\alpha}_{\pi_2}} (x,y) &=   K(K-1) f_{\hat{\alpha}}(x) \{F_{\hat{\alpha}}(y)\}^{K-2}f_{\hat{\alpha}}(y).
\end{aligned}
\end{equation}

Based on the probability theory\cite{Bertsekas}, $\Omega$ can be equally expressed as
\begin{equation}\label{Q3:6}\scriptsize
\begin{aligned}
&\Omega = \int_0^\infty \int_0^x \log_2 (\frac{1+\epsilon_M+\rho x}{1+\epsilon_M +\rho y}) f_{\hat{\alpha}_{\pi_1},\hat{\alpha}_{\pi_2}} (x,y) \text{d}y\text{d}x \cr
&=K(K-1)\int_0^\infty f_{\hat{\alpha}}(x) \underbrace{\int_0^x \log_2 (\frac{1+\epsilon_M+\rho x}{1+\epsilon_M +\rho y}) \{F_{\hat{\alpha}}(y)\}^{K-2} f_{\hat{\alpha}}(y) \text{d}y}_{{\cal{L}}_1} \text{d}x.
\end{aligned}
\end{equation}

Applying the partial integration theorem and Gauss-Chebyshev integration\cite{Hildebrand}, ${\cal{L}}_1$ can be derived as
\begin{equation}\label{Q3:7}\scriptsize
\begin{aligned}
{\cal{L}}_1 &= \int_0^x \log_2 (\frac{1+\epsilon_M+\rho x}{1+\epsilon_M +\rho y}) \text{d} {\frac{\{F_{\hat{\alpha}}(y)\}^{K-1}}{K-1}}\cr
& = \frac{\rho}{(K-1)\ln2} \int_0^x \frac{\{F_{\hat{\alpha}}(y)\}^{K-1}}{1+\epsilon_M+\rho y} \text{d}y \cr
&\approx \frac{\pi \rho}{2m(K-1)\ln2} \sum_{u=1}^{m} |\sin \frac{2u-1}{2m}\pi|\frac{x \{F_{\hat{\alpha}}(x\tau_u)\}^{K-1}}{1+\epsilon_M + \rho x\tau_u},
\end{aligned}
\end{equation}
where $\tau_u = \frac{1}{2}(\cos \frac{2u-1}{2m} \pi +1)$, and $m$ is the number of terms included in the summation, which controls the approximation accuracy.

Correspondingly, $\Omega$ can be refined as
\begin{equation}\label{Q3:9}\scriptsize
\begin{aligned}
\Omega =\frac{K \pi \rho}{2m\ln 2}  \sum_{u=1}^{m} |\sin \frac{2u-1}{2m}\pi| \underbrace{\int_0^\infty \frac{ x \{F_{\hat{\alpha}}(x\tau_u)\}^{K-1}}{1+\epsilon_M +\rho x\tau_u}  f_{\hat{\alpha}}(x) \text{d}x}_{{\cal{L}}_2}.
\end{aligned}
\end{equation}

Applying (\ref{Q3:a4}) and using the multinomial theorem, $\{F_{\hat{\alpha}}(x\tau_u)\}^{K-1}$ can be calculated as
\begin{equation}\label{Q3:10}\scriptsize
\begin{aligned}
&\{F_{\hat{\alpha}}(x\tau_u)\}^{K-1} = \sum_{r_0+r_1+ \cdots +r_n = K-1} \frac{(K-1)!}{r_0!r_1!\cdots r_n!} \cr
&\prod_{t=1}^n [-\frac{\pi}{nD}|\sin{\frac{2t-1}{2n}\pi}| x_t \exp{(- \frac{x\tau_u}{ x_t^{-\eta}-\sigma_\zeta^2} } )]^{r_t}\cr
&=\sum_{r_0+r_1+ \cdots +r_n = K-1} A(r_1,\cdots,r_n)\exp(-B(r_1,\cdots,r_n,u)x ),
\end{aligned}
\end{equation}
where $A(r_1,\cdots,r_n)$ and $B(r_1,\cdots,r_n,u)$ are shown in (\ref{thm3:add4}).

Substituting (\ref{Q3:10}) into (\ref{Q3:9}) and applying the partial integration theorem, ${\cal{L}}_2$ can be derived as
\begin{equation}\label{Q3:11}\fontsize{6.5pt}{5pt}
\begin{aligned}
&{\cal{L}}_2 = \int_0^\infty \frac{x(F_{\hat{\alpha}}(\tau_u x))^{K-1}}{1+\epsilon_M+\rho \tau_u x} \text{d}F_{\hat{\alpha}}(x) = \frac{1}{\rho \tau_u}-\int_0^\infty F_{\hat{\alpha}}(x) \text{d}\frac{x(F_{\hat{\alpha}}(\tau_u x))^{K-1}}{1+\epsilon_M+\rho \tau_u x}\cr
&= \frac{1}{\rho\tau_u} - \sum_{r_0+r_1+ \cdots +r_n = K-1} A(r_1,\cdots,r_n) \cr
&\underbrace{\int_0^\infty e^{-B(r_1,\cdots,r_n,u)x}\left[\frac{1+\epsilon_M}{(1+\epsilon_M+\rho \tau_u x)^2}-\frac{B(r_1,\cdots,r_n,u)x}{1+\epsilon_M+\rho\tau_u x} \right] F_{\hat{\alpha}}(x)\text{d}x }_{{\cal{L}}_3} .
\end{aligned}
\end{equation}

According to the $3.352$, $3.353$ and $3.194$ in\cite{Zwillinger}, we can find that
\begin{equation}\label{Q3:12}\scriptsize
\begin{aligned}
&\int_0^\infty \frac{e^{-\mu x}}{x+\xi}\text{d}x = -e^{\xi \mu} \text{Ei}(-\xi \mu), \text{if } \mu>0 \cr
& \int_0^\infty \frac{e^{-\mu x}}{(x+\xi)^2}\text{d}x = \frac{1}{\xi} +\mu e^{\xi\mu}\text{Ei}(-\xi\mu), \text{if } \mu>0\cr
&\int_0^\infty \frac{1}{(1+\xi x)^2}\text{d}x = \frac{1}{\xi},
\end{aligned}
\end{equation}
where $\text{Ei}(x) = \int_{-\infty}^x \frac{e^t}{t}\text{d}t, \text{for } x<0$.

Therefore, ${\cal{L}}_3$ can be correspondingly calculated as
\begin{equation}\label{Q3:13}\fontsize{6.6pt}{5pt}
\begin{aligned}
&{\cal{L}}_3 =  \int_0^\infty e^{-B(r_1,\cdots,r_n,u)x}\left[\frac{1+\epsilon_M}{(1+\epsilon_M+\rho \tau_u x)^2}-\frac{B(r_1,\cdots,r_n,u)x}{1+\epsilon_M+\rho\tau_u x} \right] \text{d}x\cr
&-\frac{\pi}{nD} \sum \limits_{t=1}^{n} |\sin{\frac{2t-1}{2n}\pi}| x_t  \cr
&\int_0^\infty e^{-\rho \tau_u \bar{\mu}_1 x}\left[\frac{1+\epsilon_M}{(1+\epsilon_M+\rho \tau_u x)^2}-\frac{B(r_1,\cdots,r_n,u)x}{1+\epsilon_M+\rho\tau_u x} \right] \text{d}x\cr
&=H(r_1,\cdots,r_n,u),\cr
\end{aligned}
\end{equation}
where $H(r_1,\cdots,r_n,u)$ is defined in (\ref{thm3:add4}).

Finally, substituting (\ref{Q3:13}) and (\ref{Q3:11}) into (\ref{Q3:9}), $R_{S,1}^U$ can be obtained as
\begin{equation}\small
\begin{aligned}
&R_{S,1}^U = (1-P_{out}^1) \frac{K \pi \rho}{2m\ln 2}  \sum_{u=1}^{m} |\sin \frac{2u-1}{2m}\pi|\cr
&\left\{\frac{1}{\rho\tau_u} - \sum_{r_0+r_1+ \cdots +r_n = K-1} A(r_1,\cdots,r_n) H(r_1,\cdots,r_n,u) \right\}. \cr
\end{aligned}
\end{equation}

Therefore, \emph{Theorem \ref{thm:3}} is completely proved.

\section{Derivation of $P_{out}^2$}\label{appe:4}
In order to derive $P_{out}^2$, we first need to calculate the PDF and CDF of the channel gain $\alpha_k$. Note that the distance $d$ from an arbitrary user has the PDF and CDF as $f_{d}(x) = \frac{2x}{D^2}$, $F_{d}(x) = \frac{x^2}{D^2}$. Since $d_1 \le d_2 \le \cdots \le d_K$, according to the order statistics\cite{David}, the PDF of dinstance $d_k$ is
\begin{equation}\label{Q2:a1}\footnotesize
\begin{aligned}
f_{d_k}(x) &= k \tbinom{K}{k} (F_d(x))^{k-1}(1-F_d(x))^{K-k} f_d(x)\cr
 &= 2 k \tbinom{K}{k} \frac{x^{2k-1}}{D^{2k}} (1-\frac{x^2}{D^2})^{K-k} \cr
&\mathop = \limits^{(a)} 2k \tbinom{K}{k} \sum \limits_{j=0}^{K-k} \tbinom{K-k}{j} (-1)^j \frac{x^{2(k+j)-1}}{D^{2(k+j)}},
\end{aligned}
\end{equation}
where derivation $(a)$ applies the binomial theorem. Correspondingly, the CDF of $\alpha_k$ can be evaluated as
\begin{equation}\label{Q2:a3}\footnotesize
\begin{aligned}
&F_{\alpha_k}(z) = Pr\{|g_k|^2 d_k^{-\eta} \le z\} = \int_0^D [1-\exp(-z x^\eta)] f_{d_k}(x) dx \cr
&= 1 - 2k \tbinom{K}{k} \sum \limits_{j=0}^{K-k} \tbinom{K-k}{j} \frac{(-1)^j}{D^{2(k+j)}} \int_0^D e^{-zx^\eta} x^{2(k+j)-1} dx \cr
&=1 - 2k \tbinom{K}{k} \sum \limits_{j=0}^{K-k} \tbinom{K-k}{j} \frac{(-1)^j}{D^{2(k+j)}} \frac{z^{-\frac{2(k+j)}{\eta}}}{\eta} \gamma{(\frac{2(k+j)}{\eta},zD^\eta)}.
\end{aligned}
\end{equation}

Therefore, the multicast outage probability $P_{out}^2$ can be easily calculated as
\begin{equation}\footnotesize
\begin{aligned}
&P_{out}^2 = 1 - Pr\{{\alpha}_k \ge \frac{\epsilon_M}{\rho}, \forall k \in {\cal{K}} \} \mathop = \limits^{(a)} 1- \prod_{k=1}^{K} \cr
&\left[2k \tbinom{K}{k} \sum \limits_{j=0}^{K-k} \tbinom{K-k}{j} \frac{(-1)^j}{D^{2(k+j)}} \frac{({\frac{\epsilon_M}{\rho})}^{-\frac{2(k+j)}{\eta}}}{\eta} \gamma{(\frac{2(k+j)}{\eta},\frac{\epsilon_M D^\eta}{\rho})}\right],
\end{aligned}
\end{equation}
where derivation $(a)$ results from (\ref{Q2:a3}).

Hence \emph{Theorem \ref{thm:4}} is completely proved.

\section{Derivation of $R_{S,2}^U$ for $K=2$}\label{appe:5}
As $d_1<d_2$ and the BS does not know the exact channel gain information, we assume that the unicast receiver is selected as $U_1$. Therefore, if $\alpha_1<\alpha_2$, namely $\frac{|g_1|^2}{d_1^\eta} < \frac{|g_2|^2}{d_2^\eta}$, $R_U^2 = 0$; otherwise, if $\alpha_1\ge \alpha_2\ge \frac{\epsilon_M}{\rho}$, that is $\frac{|g_1|^2}{d_1^\eta} \ge \frac{|g_2|^2}{d_2^\eta}\ge \frac{\epsilon_M}{\rho}$, then $\theta^*_U = \frac{\alpha_2 - \frac{\epsilon_M}{\rho}}{\alpha_2 (1+\epsilon_M)}$, and $R_U^2 = \log_2 (1+\rho \alpha_1 \theta^*_U) - \log_2 (1+\rho \alpha_2 \theta^*_U)$. The exact closed-form expression of $R_{S,2}^U$ is difficult to obtain. Here an approximate expression under high
SNR (large $\rho$) is derived. Note that when $\rho$ is large enough, $\theta_U^* \approx \frac{1}{1+\epsilon_M}$. Similarly, $R_{S,2}^U$ can be calculated as
\begin{equation}\label{Q2:9}\fontsize{7pt}{5pt}
\begin{aligned}
R_{S,2}^U = \text{E}\left[\log_2 (\nu + \rho \frac{|g_1|^2}{d_1^\eta}) -\log_2 (\nu + \rho \frac{ |g_2|^2}{d_2^\eta}) | \frac{|g_1|^2}{d_1^\eta} \ge \frac{|g_2|^2}{d_2^\eta}\ge \frac{\epsilon_M}{\rho} \right].
\end{aligned}
\end{equation}

As discussed above, it is easy to see that the PDF of $|g_k|^2$ follows exponential distribution with zero mean and unit variance, namely $f_{|g_k|^2}(x) \sim \exp(-x)$. Furthermore, according to the $2.2.1$ in\cite{David} (described in (\ref{Q3:1})), the joint PDF of $d_1$ and $d_2$ is $f_{d_1,d_2}(x,y) = \frac{8xy}{D^4},0<x<y<D$. Hence, $R_{S,2}^U$ can be derived as
\begin{equation}\label{Q2:10}\fontsize{6.5pt}{5pt}
\begin{aligned}
&R_{S,2}^U = \frac{8}{D^4 \ln2} \int_{x=0}^{D} \int_{u=0}^\infty \int_{y=x}^D \int_{v=\frac{\epsilon_M y^\eta}{\rho}}^{\frac{uy^\eta}{x^\eta}} \cr
&\left[\ln (\nu + \rho \frac{u}{x^\eta}) -\ln (\nu + \rho \frac{ v}{y^\eta})\right] xy e^{-u} e^{-v} \text{d}v \text{d}y \text{d}u \text{d}x \cr
& = \frac{8}{D^4 \ln2} \underbrace{\int_{x=0}^{D} \int_{u=0}^\infty \int_{y=x}^D \int_{v=\frac{\epsilon_M y^\eta}{\rho}}^{\frac{uy^\eta}{x^\eta}} \ln (\nu + \rho \frac{u}{x^\eta})  xy e^{-u} e^{-v} \text{d}v \text{d}y \text{d}u \text{d}x}_{{\cal{L}}_4}  \cr
&- \frac{8}{D^4 \ln2} \underbrace{\int_{y=0}^{D} \int_{v=\frac{\epsilon_M y^\eta}{\rho}}^\infty \int_{x=0}^y \int_{u=\frac{vx^\eta}{y^\eta}}^{\infty}\ln (\nu + \rho \frac{ v}{y^\eta}) xy e^{-u} e^{-v} \text{d}u \text{d}x \text{d}v \text{d}y}_{{\cal{L}}_5}.
\end{aligned}
\end{equation}

As to ${\cal{L}}_4$, ${\cal{W}}_1 = \int_{y=x}^D \int_{v=\frac{\epsilon_M y^\eta}{\rho}}^{\frac{uy^\eta}{x^\eta}} ye^{-v}\text{d}v\text{d}y$ can be first refined as
\begin{equation}\label{Q2:11}\footnotesize
\begin{aligned}
{\cal{W}}_1 &= \int_{y=x}^D y(e^{-\frac{\epsilon_M y^\eta}{\rho}}-e^{-\frac{uy^\eta}{x^\eta}})\text{d}y = -\frac{x^2}{\eta u^{\frac{2}{\eta}}} \gamma(\frac{2}{\eta},\frac{uD^\eta}{x^\eta}) \cr
&+ \frac{x^2}{\eta u^{\frac{2}{\eta}}} \gamma(\frac{2}{\eta},u)+\frac{1}{\eta}(\frac{\rho}{\epsilon_M})^{\frac{2}{\eta}} \left[\gamma(\frac{2}{\eta},\frac{\epsilon_M D^\eta}{\rho}) -\gamma(\frac{2}{\eta},\frac{\epsilon_M x^\eta}{\rho}) \right].
\end{aligned}
\end{equation}
Note that ${\cal{L}}_4$ needs to calculate double integral over ${\cal{W}}_1$. In order to obtain more insights on ${\cal{L}}_4$, by applying the Gauss-Chebyshev Integration to the lower incomplete gamma function, ${\cal{W}}_1$ can be approximated as
\begin{equation}\label{Q2:12}\footnotesize
\begin{aligned}
{\cal{W}}_1 & \approx \frac{\pi}{2l\eta}\sum_{i=1}^l |\sin(\frac{2i-1}{2l}\pi)| \kappa_i^{\frac{2}{\eta}-1} \cr
&\left[D^2 e^{-\frac{\epsilon_M D^\eta}{\rho}\kappa_i}- x^2 e^{-\frac{\epsilon_M x^\eta}{\rho}\kappa_i}+x^2e^{-u\kappa_i}-D^2e^{-\frac{uD^\eta}{x^\eta}\kappa_i} \right],
\end{aligned}
\end{equation}
where $\kappa_i = \frac{1}{2}(1+\cos(\frac{2i-1}{2l}\pi))$, and $l$ is the number of terms included in the summation, which controls the approximation accuracy.

Next, ${\cal{W}}_2 = \int_{u=0}^\infty \ln (\nu + \rho \frac{u}{x^\eta}) e^{-u} {\cal{W}}_1  \text{d}u$ can be derived as
\begin{equation}\label{Q2:13}\footnotesize
\begin{aligned}
{\cal{W}}_2 &=\frac{\pi}{2l\eta}\sum_{i=1}^l |\sin(\frac{2i-1}{2l}\pi)| \kappa_i^{\frac{2}{\eta}-1} [\int_0^\infty (D^2 e^{-\frac{\epsilon_M D^\eta}{\rho}\kappa_i}\cr
&- x^2 e^{-\frac{\epsilon_M x^\eta}{\rho}\kappa_i}+x^2e^{-u\kappa_i}-D^2e^{-\frac{uD^\eta}{x^\eta}\kappa_i}) \ln (\nu + \rho \frac{u}{x^\eta}) e^{-u} \text{d}u].
\end{aligned}
\end{equation}
According to $4.337$ in\cite{Zwillinger}, we can find that $\int_0^\infty \ln (1+\xi x) e^{-\mu x} \text{d} x = -\frac{1}{\mu} e^{\frac{\mu}{\xi}}\text{Ei}(-\frac{\mu}{\xi}), \text{for } \mu>0$. Therefore,
\begin{equation}\label{Q2:14}\footnotesize
\begin{aligned}
\int_0^\infty \ln (\nu+ax) e^{-\mu x} \text{d} x = \frac{\ln \nu}{\mu} -\frac{1}{\mu} e^{\frac{\nu \mu}{a}}\text{Ei}(-\frac{\nu \mu}{a}), \text{for } \mu>0.
\end{aligned}
\end{equation}
Correspondingly, ${\cal{W}}_2$ can be calculated as
\begin{equation}\label{Q2:15}\fontsize{7pt}{5pt}
\begin{aligned}
&{\cal{W}}_2 =  \frac{\pi}{2l\eta}\sum_{i=1}^l |\sin(\frac{2i-1}{2l}\pi)| \kappa_i^{\frac{2}{\eta}-1}\{(D^2 e^{-\frac{\epsilon_M D^\eta}{\rho}\kappa_i}- x^2 e^{-\frac{\epsilon_M x^\eta}{\rho}\kappa_i})\cr
&\times [\ln \nu-e^{\frac{\nu x^\eta}{\rho}}\text{Ei}(-\frac{\nu x^\eta}{\rho}) ] +x^2[-\frac{1}{\kappa_i+1}e^{\frac{\nu(\kappa_i+1)x^\eta}{\rho}}\text{Ei}(-\frac{\nu(\kappa_i+1)x^\eta}{\rho})\cr
&+\frac{\ln \nu}{\kappa_i+1}]-D^2[- \frac{x^\eta}{D^\eta\kappa_i+x^\eta}e^{\frac{\nu(D^\eta\kappa_i+x^\eta)}{\rho}}\text{Ei}(-\frac{\nu(D^\eta\kappa_i+x^\eta)}{\rho})\cr
&+\frac{x^\eta \ln \nu}{D^\eta\kappa_i+x^\eta} ]  \}.
\end{aligned}
\end{equation}
According to the $3.381$ in\cite{Zwillinger}, namely
\begin{equation}\label{Q2:151}\footnotesize
\begin{aligned}
\int_0^u x^m e^{-\xi x^n}\text{d}x = \frac{\gamma(\frac{m+1}{n},\xi u^n)}{n\xi^{\frac{m+1}{n}}}, \text{for } u>0,n>0,\xi>0,
\end{aligned}
\end{equation}
${\cal{L}}_4 = \int_0^D x{\cal{W}}_2  \text{d}x$ can be calculated as
\begin{equation}\label{Q2:16}\footnotesize
\begin{aligned}
{\cal{L}}_4& =\frac{\pi}{2l\eta}\sum_{i=1}^l |\sin(\frac{2i-1}{2l}\pi)| \kappa_i^{\frac{2}{\eta}-1} J(i) \ln \nu \cr
&+\frac{\pi}{2l\eta}\sum_{i=1}^l |\sin(\frac{2i-1}{2l}\pi)| \kappa_i^{\frac{2}{\eta}-1}\underbrace{\int_0^D x J_1(i,x)\text{d}x}_{{\cal{W}}_3},
\end{aligned}
\end{equation}
where
\begin{equation}\label{Q2:161}\footnotesize
\begin{aligned}
&J_1(i,x) = -e^{\frac{\nu x^\eta}{\rho}}\text{Ei}(-\frac{\nu x^\eta}{\rho})(D^2 e^{-\frac{\epsilon_M D^\eta}{\rho}\kappa_i}- x^2 e^{-\frac{\epsilon_M x^\eta}{\rho}\kappa_i})\cr
&-\frac{x^2}{\kappa_i+1}e^{\frac{\nu (\kappa_i+1)x^\eta}{\rho}}\text{Ei}(-\frac{\nu(\kappa_i+1)x^\eta}{\rho})\cr
&-D^2[\frac{x^\eta \ln \nu}{D^\eta\kappa_i+x^\eta} - \frac{x^\eta}{D^\eta\kappa_i+x^\eta}e^{\frac{\nu(D^\eta\kappa_i+x^\eta)}{\rho}}\text{Ei}(-\frac{\nu(D^\eta\kappa_i+x^\eta)}{\rho})].
\end{aligned}
\end{equation}

By using the Gaussian-Chebyshev integration, ${\cal{W}}_3$ can be approximated as
\begin{equation}\label{Q2:17}
\begin{aligned}
{\cal{W}}_3 \approx \frac{D\pi}{2q}\sum_{j=1}^q |\sin(\frac{2j-1}{2q}\pi)| x_j J(i,x_j),
\end{aligned}
\end{equation}
where $x_j = \frac{D}{2}(1+\cos(\frac{2j-1}{2q}\pi))$, and $q$ is the number of terms included in the summation, which controls the approximation accuracy.

Therefore, ${\cal{L}}_4 $ can be obtained as
\begin{equation}\label{Q2:18}\footnotesize
\begin{aligned}
{\cal{L}}_4 &= \frac{\pi}{2l\eta}\sum_{i=1}^l |\sin(\frac{2i-1}{2l}\pi)| \kappa_i^{\frac{2}{\eta}-1}\{ J(i) \ln \nu \cr
&+ \frac{D\pi}{2q}\sum_{j=1}^q |\sin(\frac{2j-1}{2q}\pi)| x_j J(i,x_j) \}.
\end{aligned}
\end{equation}

As to ${\cal{L}}_5$, similarly ${\cal{W}}_4 = \int_{x=0}^y \int_{u=\frac{vx^\eta}{y^\eta}}^\infty xe^{-u}\text{d}u\text{d}x$ can be first calculated as
\begin{equation}\label{Q2:19}\footnotesize
\begin{aligned}
{\cal{W}}_4 &= \int_{0}^y x e^{-\frac{vx^\eta}{y^\eta}} \text{d}x = \frac{y^2}{v^{\frac{2}{\eta}}\eta} \gamma(\frac{2}{\eta},v) \cr
& \mathop \approx \limits^{(a)} \frac{\pi}{2l\eta} \sum_{i=1}^l |\sin(\frac{2i-1}{2l}\pi)| \kappa_i^{\frac{2}{\eta}-1}y^2 e^{-\kappa_i v},
\end{aligned}
\end{equation}
where derivation $(a)$ results from the Gauss-Chebyshev integration.

Then ${\cal{W}}_5= \int_{\frac{\epsilon_M y^\eta}{\rho}}^\infty \ln(\nu+\rho\frac{v}{y^\eta}){\cal{W}}_4 e^{-v}\text{d}v$  can be derived as

\begin{equation}\label{Q2:20}\footnotesize
\begin{aligned}
{\cal{W}}_5 &= \frac{\pi}{2l\eta} \sum_{i=1}^l |\sin(\frac{2i-1}{2l}\pi)| \kappa_i^{\frac{2}{\eta}-1}y^2 \int_{\frac{\epsilon_M y^\eta}{\rho}}^\infty \ln(\nu+\rho\frac{v}{y^\eta}) e^{-(\kappa_i+1) v} \text{d}v \cr
&= \frac{\pi}{2l\eta} \sum_{i=1}^l |\sin(\frac{2i-1}{2l}\pi)| \kappa_i^{\frac{2}{\eta}-1}y^2 e^{-\frac{\epsilon_M y^\eta (\kappa_i+1)}{\rho}}\cr
&\int_0^\infty \ln(\nu+\epsilon_M+\frac{\rho}{y^\eta}t)e^{-(\kappa_i+1)t}\text{d}t \cr
&= \frac{\pi}{2l\eta} \sum_{i=1}^l |\sin(\frac{2i-1}{2l}\pi)| \kappa_i^{\frac{2}{\eta}-1}y^2 e^{-\frac{\epsilon_M y^\eta (\kappa_i+1)}{\rho}}[\frac{\ln{( \nu+\epsilon_M)}}{\kappa_i+1}\cr
&-\frac{1}{\kappa_i+1}e^{\frac{(\nu+\epsilon_M)(\kappa_i+1)y^\eta}{\rho}} \text{Ei}(-\frac{(\nu+\epsilon_M)(\kappa_i+1)y^\eta}{\rho})].
\end{aligned}
\end{equation}

Hence, based on (\ref{Q2:151}) and applying the Gauss-Chebyshev integration, ${\cal{L}}_5 = \int_0^D y{\cal{W}}_5\text{d}y$ can be calculated as
\begin{equation}\label{Q2:201}\footnotesize
\begin{aligned}
{\cal{L}}_5 &= \frac{\pi}{2l\eta} \sum_{i=1}^l |\sin(\frac{2i-1}{2l}\pi)| \kappa_i^{\frac{2}{\eta}-1} \{\frac{ \ln(\nu+\epsilon_M)\gamma(\frac{4}{\eta},\frac{(\kappa_i+1)\epsilon_M D^\eta}{\rho}) }{\eta(\kappa_i+1)[\frac{(\kappa_i+1)\epsilon_M}{\rho}]^{4/\eta}} \cr &-\frac{D\pi}{2q}\sum_{j=1}^q |\sin(\frac{2j-1}{2q}\pi)| x_j J_2(i,x_j)\},
\end{aligned}
\end{equation}
where $J_2(i,x)= \frac{x^2}{\kappa_i+1}e^{\frac{(\kappa_i+1)x^\eta \nu}{\rho}}\text{Ei}(-\frac{(\nu+\epsilon_M)(\kappa_i+1)x^\eta}{\rho})$.

Finally, substituting (\ref{Q2:18}) and (\ref{Q2:201}) into $R_{S,2}^U = \frac{8}{D^4\ln 2}({{\cal{L}}_4}-{{\cal{L}}_5})$, (\ref{Q2:21}) can be correspondingly achieved, and \emph{Theorem \ref{thm:5}} is completely proved.

% if have a single appendix:
%\appendix[Proof of the Zonklar Equations]
% or
%\appendix  % for no appendix heading
% do not use \section anymore after \appendix, only \section*
% is possibly needed

% use appendices with more than one appendix
% then use \section to start each appendix
% you must declare a \section before using any
% \subsection or using \label (\appendices by itself
% starts a section numbered zero.)
%

\iffalse
\appendices
\section{Proof of the First Zonklar Equation}
Appendix one text goes here.

% you can choose not to have a title for an appendix
% if you want by leaving the argument blank
\section{}
Appendix two text goes here.
\fi

% use section* for acknowledgment
\iffalse
\section*{Acknowledgment}

The authors would like to thank...
\fi

% Can use something like this to put references on a page
% by themselves when using endfloat and the captionsoff option.
\ifCLASSOPTIONcaptionsoff
  \newpage
\fi

\end{document}